\documentclass{amsart}

\AtBeginDocument{%
  \providecommand\BibTeX{{%
    \normalfont B\kern-0.5em{\scshape i\kern-0.25em b}\kern-0.8em\TeX}}}
    
\usepackage{type1cm}

 \usepackage[T1]{fontenc}
\usepackage{booktabs}
\usepackage{stmaryrd}
\usepackage{fdsymbol}
\usepackage{enumitem}
\usepackage{xspace}
\usepackage{url}	
\usepackage{xcolor}
\usepackage{stackengine}
\usepackage{graphicx}
\usepackage{subfigure}

\newcommand{\LTL}{\mathrm{LTL}}
\newcommand{\op}{\mathbf{\mathrm{op}}}
\newcommand{\always}{\Box}
\newcommand{\eventu}{\Diamond}
\newcommand{\Next}{\bigcirc}
\renewcommand{\implies}{\Rightarrow}
\newcommand{\until}{\ensuremath{\mathrel{\mathcal{U}}}\xspace}
\newcommand{\release}{\ensuremath{\mathrel{\mathcal{R}}}\xspace}
\newcommand{\rLTL}{\mathrm{rLTL}}
\newcommand{\ralways}{\boxdot}
\newcommand{\rimplies}{\Rrightarrow}
\newcommand{\reventu}{\Diamonddot}
\newcommand{\untildot}{\ensuremath{\mathrel{\topinset{$\cdot$}{$\until$}{2pt}{1pt}}}\xspace}
\newcommand{\releasedot}{\ensuremath{\mathrel{\topinset{$\cdot$}{$\release$}{1.2pt}{0.6pt}}}\xspace}
\newcommand{\nextdot}{\odot}
\newcommand{\ltl}{\mathrm{ltl}}
\newcommand{\thefragment}{\rLTL_{\backslash\{\rimplies\}}(\mathcal{P})}
\newcommand{\thelargerfragment}{\overline{\rLTL}_{\backslash\{\rimplies\}}(\mathcal{P})}
\usepackage[ruled,vlined]{algorithm2e}

\newcommand{\suf}[1]{{#1..}}
\newcommand{\true}{\ensuremath{\mathit{true}}\xspace}
\newcommand{\false}{\ensuremath{\mathit{false}}\xspace}
\newcommand{\A}{\mathcal{A}}
\newcommand{\G}{\mathcal{G}}

\renewcommand{\P}{\mathcal{P}}
\newcommand{\K}{\mathcal{K}}
\newcommand{\T}{\mathcal{T}}

\newcommand{\B}{{\mathbb B}}
\newcommand{\bbit}{\B_5}

\newcommand{\N}{{\mathbb{N}}}
\newcommand{\cl}{\mathrm{cl}}
\newcommand{\bigO}{\mathcal{O}}
\newtheorem{theorem}{Theorem}[section]
\newtheorem{definition}{Definition}
\newtheorem{remark}{Remark}
\newtheorem{lemma}[theorem]{Lemma}
\newtheorem{proposition}[theorem]{Proposition}
\newtheorem{corollary}[theorem]{Corollary}
\newtheorem{example}{Example}
\newtheorem{problem}{Problem}

\usepackage{soul}
\newcommand{\specialcell}[2][c]{%
  \begin{tabular}[#1]{@{}c@{}}#2\end{tabular}}
\renewcommand{\reventu}{\diamonddot}

\begin{document}

\title[Being correct is not enough: efficient verification using robust linear temporal logic]{Being correct is not enough: efficient verification using\\robust linear temporal logic}

\author{Tzanis Anevlavis}
\address{Department of Electrical Engineering\\
University of California at Los Angeles,
Los Angeles, CA 90095}
\email{t.anevlavis@ucla.edu}
\email{tabuada@ucla.edu}

\author{Matthew Philippe}
\address{Universit\`e catholique de Louvain\\
Louvain, Belgium}
\email{matthew.philippe@uclouvain.be}

\author{Daniel Neider}
\address{Max Planck Institute for Software Systems\\
Kaiserslauten, Germany}
\email{neider@mpi-sws.org}

\author{Paulo Tabuada}

\thanks{This work was partially supported by the NSF grant 1645824, the CONIX Research Center, one of six centers in JUMP, a Semiconductor Research Corporation (SRC) program sponsored by DARPA, and the Deutsche Forschungsgemeinschaft (DFG, German Research Foundation) grant number 434592664.}

\renewcommand{\shortauthors}{Anevlavis, et al.}

\begin{abstract}
While most approaches in formal methods address system correctness, ensuring robustness has remained a challenge. In this paper we present and study the logic rLTL which provides a means to formally reason about both correctness and robustness in system design. Furthermore, we identify a large fragment of rLTL for which the verification problem can be efficiently solved, i.e., verification can be done by using an automaton, recognizing the behaviors described by the rLTL formula $\varphi$, of size at most $\bigO \left( 3^{ |\varphi|} \right)$, where $|\varphi|$ is the length of $\varphi$. This result improves upon the previously known bound of $\bigO \left(5^{|\varphi|} \right)$ for rLTL verification and is closer to the LTL bound of $\bigO \left( 2^{|\varphi|} \right)$.
The usefulness of this fragment is demonstrated by a number of case studies showing its practical significance in terms of expressiveness, the ability to describe robustness, and the fine-grained information that rLTL brings to the process of system verification. Moreover, these advantages come at a low computational overhead with respect to LTL verification. 
\end{abstract}

\maketitle

\section{Introduction}
\label{sec:intro}
As Cyber-Physical Systems (CPS) inevitably become increasingly complex, the ability to completely guarantee correctness of their design and implementation via exhaustive testing fades. Moreover, almost every aspect of contemporary life is becoming intertwined with CPSs, such smart grids, smart cities, mobility on demand and autonomous vehicles, and even medical devices. Consequently, in an attempt to reduce design errors, formal methods have been investigated to support modeling and verification of CPS and, in particular, of its reactive components. 

Most work in formal methods has focused on system correctness, i.e., in ensuring that systems are guaranteed to meet their design specifications. We argue that correctness is necessary, but not sufficient for a good design when a reactive system interacts with an ever-changing uncontrolled environment. To illustrate this point, just consider the correctness specifications for open reactive systems, which are typically written in the form of an implication:
\begin{equation}
\label{Eq:Specification}
	\varphi\Rightarrow \psi,
\end{equation}
where $\varphi$ is an environment assumption and $\psi$ is a system guarantee. In Linear Temporal Logic (LTL) the implication in \eqref{Eq:Specification} is equivalent to $\neg\varphi\lor \psi$, and, ergo, whenever the assumption $\varphi$ is violated the above specification yields no information on the guarantee. In other words, the system can behave arbitrarily. 
Thus, in addition to correctness, systems should also be designed to be robust, i.e., small deviations from the assumptions made at design time should lead to, at most, small violations of the design specifications. While it is hard to dispute that design time assumptions may not hold in the environments where systems will actually be deployed, since these may not be completely known at design time or may be evolving over time, how to formally describe robustness is a question that has not received enough attention despite the recent efforts described in Section~\ref{subsec:litrev}. In this paper, we address this problem by studying a recently developed logic, termed robust Linear-time Temporal Logic and abbreviated as rLTL, that allows to specify robustness. Its syntax closely mirrors that of LTL to lower the barriers to its adoption. Its semantics, however, is different in many regards. In particular, it is a many-valued logic so that one can reason about the different ways in which assumptions and guarantees can be violated. 

To shed more light into the mechanics of rLTL, consider the LTL formula $\always p$ with $p$ an atomic proposition. There is \emph{only one way} in which this formula is satisfied, namely when $p$ holds at every time step. In contrast, there are \emph{several ways} in which this formula can be violated over an infinite trace: (1) the worst possible violation occurs when $p$ fails to hold at every time step; (2) a slightly better scenario is where $p$ holds for at most finitely many time instants; (3) better yet would be that $p$ holds at infinitely many instants, while still failing to hold at infinitely many instants; (4) finally, among all the possible ways in which $\always p$ can be violated, the most preferable case would be the one where $p$ fails to hold for at most finitely many time instants. The semantics of rLTL is exactly designed to distinguish between these different ways. 

In preliminary work~\cite{tabuada2016rLTL}, we introduced a fragment of rLTL that only contained the always and eventually operators. We  showed, in that restricted context, that we can decide if a system satisfies an rLTL formula $\varphi$ by using an automaton with size $\bigO \left(5^{|\varphi|} \right)$, where $| \varphi |$ denotes the length of $\varphi$. The corresponding automaton for LTL has size $\bigO \left( 2^{|\varphi|} \right)$ and the change in the base of the exponential follows from the fact that LTL is a $2$-valued logic, whereas rLTL is $5$-valued. In this paper we offer a fragment of rLTL that also includes the next, until, and release operators, while placing a syntactic restriction on the antecedent of nested implications. For the proposed fragment of rLTL the verification problem can be solved more efficiently, i.e., by using an automaton of size $\bigO \left( 2^{ |\varphi| - \kappa(\varphi)} 3^{\kappa(\varphi)} \right)$, where $\kappa(\varphi)$ measures the number of unique subformulae of $\varphi$ that contain always and release operators. Construction of smaller automata is achieved by using temporal testers and exploiting properties of the proposed fragment. In particular, this provides the upper bound $\bigO \left( 3^{ |\varphi|} \right)$ on the size of this automaton which is closer to the LTL bound. In prior work~\cite{anevlavis2018rLTL} we achieve the same complexity bound, but for a smaller fragment of rLTL. The fragment studied in~\cite{anevlavis2018rLTL} does not contain the release operator and allows at most one implication operator at the outermost level. Consequently, the fragment proposed here is substantially larger as evidenced by the second of our case studies, detailed next.

To illustrate the usefulness of rLTL and the proposed fragment, we offer several case studies that demonstrate: (1) how rLTL can identify a non-robust system, whereas LTL cannot, as it does not provide information about the guarantee of an implication when the environment deviates from the modeling assumptions; (2) how the proposed fragment contains the most important reactivity patterns~\cite{dwyer1999patterns}; and (3) how the five truth values provide insightful information when a specification is violated, which can be useful for a designer seeking to improve the designed system and/or the specification. Moreover, the computational overhead associated with rLTL model-checking is relatively low with respect to LTL model-checking, and also rLTL model-checking, within the proposed fragment, scales similarly to LTL model-checking with respect to the size of the model-checked formula, as shown by our experiments.

\subsection{In a labyrinth of robustness}
\label{subsec:litrev}
A number of efforts has been made in order to express the ``correct'' notion of robustness with regards to cyber-physical systems in formal methods. In this section, we present an extensive, but not exhaustive, review of various formalizations of robustness. We begin by a series of approaches, which require the designer to provide information in addition to the desired specification. 

In~\cite{bloem2014synthrobsys}, two quantitative robustness concepts are combined in a common framework for robust synthesis. 
The first one, robustness for safety, looks at how often the assumptions and the guarantees are violated, and asks for their ratio to be bounded by $k\in\N$ (k-robustness). Counting is achieved through error functions provided by the designer. The second concept, robustness for liveness, considers specifications of the form $\land_{i\in I}\,\eventu\always p_i \implies \land_{j\in J}\,\eventu\always q_j$, where $p_i$, $q_j$ are atomic propositions, and then compares the number of violated assumptions to the number of violated guarantees. The rLTL semantics, even though being able to distinguish between the different ways in which a specification can be violated, does not distinguish between the violation of one assumption from the violation of multiple assumptions. 
Hence, the second approach cannot be compared to the one proposed here. Furthermore, we make no distinctions between safety and liveness properties. 

Moreover, in~\cite{bloem2019synthreacsys}, a different framework for robust synthesis is proposed, that does not encompass the one above. Different notions of robustness are considered, e.g., a robust system satisfies a guarantee, even though a finite number, or even all, of the inputs are hidden/misread, or even though the assumption is violated finitely/infinitely often. Many of the considered notions are incorporated in rLTL, and in fact our definition of robustness allows systems to satisfy weaker guarantees, whenever the assumptions are also weakened, which is more general.  Nonetheless, we cannot compare with the notions of robustness in~\cite{bloem2019synthreacsys} that count the number of violations, since the rLTL semantics distinguishes only between zero, finite, and infinite violations of a specification.

In the intriguing work of~\cite{rodionova2016logicfiltering}, a link between both MTL/LTL, and Linear Time-Invariant (LTI) filtering is established. Specifically, it is shown that LTI filtering corresponds to MTL if addition and multiplication are interpreted as max and min, and if true and false are interpreted as one and zero. By using different filtering kernels, one expresses weaker or stronger interpretations of the same formula. However, this burdens the designer both with the choice of kernels and the use multiple semantics to reason about how weakening the assumptions leads to weakening the guarantees. 

Contrary to all the approaches described so far, which require robustness metrics to be provided by the designer, when working with rLTL the designer only needs to provide the desired specification and no other information. Hence, we ease the designer's effort since it is not always clear which quantitative metric leads to the desired qualitative behavior.

Another interpretation of robustness is provided in [48] as the difference between the number of steps violating the guarantee and the number of steps violating the assumption of a reactive specification. A more robust reactive system produces smaller such differences, and robustness is evaluated quantitatively as a mean-payoff objective, averaged over all executions of the system. Alternatively, the discounted-sum can be used as a quantitative objective, which offers convergence properties over infinite runs. Moreover, this objective has been shown to pair well with qualitative LTL constraints, both in the reinforcement learning domain~\cite{wen2015rlwithltl}
and to obtain sound and efficient automata-based algorithms for quantitative reasoning ~\cite{bansal2021satisficing}.
Again, such approaches are not comparable to ours since rLTL distinguishes only between zero, finite, and infinite violations of a specification. 

In the domain of software systems,~\cite{zhang2020behavrobsoftsys} defines robustness as the largest set of deviating environmental behaviors under which the system still guarantees a desired property. Therefore, robustness is defined, with respect to a property, as the set of all deviations under which a system continues to satisfy that property. Although this work focuses on computing robustness, rather than characterizing it, it is possible that certain temporal deviations could be expressed in rLTL. Additional noteworthy works, although incomparable with the methods described here, are~\cite{chaudhuri2010contanalysisprograms} and~\cite{majumdar2009symbrobanalysis}, which consider continuity properties of software expressed by the requirement that a deviation in a program's input causes a proportional deviation in its output. Although natural, these notions of robustness only apply to the Turing model of computation and not to the reactive model of computation employed in this paper. 

A plethora of works exists regarding robustness of specifications when reasoning over real-valued, continuous-time signals, with the most prominent being~\cite{fainekos2009robtemplogcontsignals},~\cite{donze2010robtemplogrealval}. In these works, no discussion of the specific choices made when crafting the many-valued semantics is provided. Interestingly, though, the notion of ``time robustness'' in~\cite{donze2010robtemplogrealval} is close to the one of rLTL in the sense that it measures the time needed for the truth value of a formula to change. Nevertheless, in this line of work robustness is derived from the real-valued nature of the signals, whereas in rLTL, we reason over the more classical setting of discrete-time and Boolean valued signals, with robustness derived from the temporal evolution of these signals. Consequently, the works of~\cite{fainekos2009robtemplogcontsignals},~\cite{donze2010robtemplogrealval}, and their extensions can be considered of orthogonal and complementary nature to ours. 

Another relevant approach of multi-valued extensions of LTL is found in~\cite{almagor2016quality}. This work introduces two quantitative extensions of LTL, one by propositional quality operators termed $\LTL[\mathcal{F}]$, parameterized by a set $\mathcal{F}$ of functions over $[0,1]$,
and one by discounting operators termed $\LTL^{disc}[\mathcal{D}]$, parameterized by a set $\mathcal{D}$ of discounting functions. 
Both logics employ a many-valued variant of LTL to reason about quality, and the satisfaction value of a specification is a number in $[0, 1]$, which describes the quality of the satisfaction. The use of a many-valued semantics in the context of quality is as natural as in the context of robustness. In fact, it was shown in~\cite{tabuada2016rLTL} that by dualizing the semantics of rLTL in a specific sense we obtain a logic that is adequate to reason about quality. 
Nevertheless, there are strong conceptual differences between the approach taken in this paper and the approach in~\cite{almagor2016quality}. First, our notion of robustness or quality is intrinsic to the logic, while the approach in~\cite{almagor2016quality} requires the designer to provide their own interpretation in the form of the sets $\mathcal{F}$ or $\mathcal{D}$ of functions that parameterize the logic. Second, there are several choices to define the logical connectives on the interval $[0,1]$. As an illustration for the latter, note that there are three commonly used conjunctions: \L{}ukasiewicz's conjunction $a\land b=\max\{0,a+b-1\}$, G\"odel's conjunction $a\land b=\min\{a,b\}$, and the product of real numbers $a\land b=a \cdot b$ also known as Goguen's conjunction. Moreover, each such choice leads to a different notion of implication via residuation. Whether G\"odel's conjunction, used in~\cite{almagor2016quality}, is the most adequate to formalize quality is a question not addressed in~\cite{almagor2016quality}. In contrast, we carefully discuss and motivate all the choices made when defining the semantics of rLTL with robustness considerations. 

The last body of work related to the contents of this paper is~\cite{kupferman2007latticeautomata, almagor2014latticedsynthesis} on lattice automata and lattice LTL. The syntax of lattice LTL is similar to the syntax of LTL except that atomic propositions assume values on a finite lattice (which has to satisfy further restrictions such as being distributive). Although both lattice LTL as well as rLTL are many-valued logics, lattice LTL derives its many-valued character from the atomic propositions. In contrast, atomic propositions in rLTL are interpreted classically (i.e., they only assume two truth values). Therefore, the many-valued character of rLTL arises from the temporal evolution of the atomic propositions and not from the nature of the atomic propositions or their interpretation. In fact, if we only allow two truth values for the atomic propositions in lattice LTL, as is the case for rLTL, lattice LTL degenerates into LTL. Hence, these two logics capture orthogonal considerations, and results on lattice LTL and lattice automata do not shed light on how to address similar problems for rLTL.

\subsection{Beyond rLTL verification}
\label{subsec:morerltl}
This paper studies the verification problem for rLTL. As we have already mentioned, the rLTL semantics allows for specifying robustness and is able to distinguish between the different ways in which a specification is violated in a qualitative manner, that is, between zero, finitely, and infinitely many violations. Nonetheless, rLTL does not distinguish between different numbers of assumptions being violated and, hence, cannot count over bounded time segments where a specification is defined as a conjunction of assumptions, each corresponding to a time step. 

In addition to the verification problem, other problems for rLTL have been studied in the literature. The work in~\cite{tabuada2016rLTL} studies the synthesis problem for a restricted fragment that contains only the always and eventually operators, while~\cite{tabuada2015rLTLarxiv} extends the same results to full rLTL. Both works consider the environment to be antagonistic, which is not very realistic and leads to suboptimal controllers. This issue is addressed in~\cite{nayak2021rltlgames} by introducing adaptive strategies, which are not more complex than the classical ones, and take advantage of the environment making bad moves. Furthermore, the runtime monitoring problem, i.e., checking properties of infinite words based on a given finite prefix, is studied in~\cite{mascle202rltlmonitoring} in the context of rLTL. Finally, different shortcomings of LTL, other than the lack of robustness, such as the limited expressiveness and the lack of quantitative features, have been addressed by other extensions like Linear Dynamic Logic~\cite{vardi2011riseandfall} and Prompt-LTL~\cite{kupferman2009promptness} respectively. While the above logics and rLTL address each shortcoming separately, the work in~\cite{neider2019picktwo} shows how to combine any two of the aforementioned extensions and at the same time do not incur any additional complexity overhead from merging the logics.

\subsection{Outline of the paper}
The outline of the paper is as follows: in Section~\ref{sec:prelims}, the basic definitions of LTL and results on the LTL model-checking problem are reviewed. Following up in Section~\ref{sec:rLTLsyntaxsem}, the syntax of rLTL is presented, its semantics is thoroughly derived, and the necessity of the $5$ truth values is justified. Subsequently, Section~\ref{subsec:rltl2ltl} discusses the relationship between rLTL and LTL, and provides translations between the two. Section~\ref{sec:rLTLmodelchecking} defines the rLTL model-checking problem and settles its decidability. After that, the fragment for efficient rLTL model-checking is identified in Section~\ref{sec:augmentedfragment}, and refined complexity bounds for the rLTL model-checking are derived. A comprehensive case study on rLTL model-checking is conducted in Section~\ref{sec:casestudies}. Finally, concluding remarks are found in Section~\ref{sec:conclusion}.

\section{Preliminaries}
\label{sec:prelims}
\subsection{Notation}
Let $\N = \{0, 1, \ldots\}$ be the set of natural numbers and $\B = \{0, 1\}$ the set of Boolean values with $0$ interpreted as \false and $1$ interpreted as \true. For a set $A$, let $2^A$ be the \emph{powerset} of $A$, i.e., the set of all subsets of $A$, let $A^\omega$ be the set of all \emph{infinite sequences} of elements of $A$, and let $\mathrm{card}(A)$ denote the cardinality of $A$. An \emph{alphabet} $\Sigma$ is a finite, nonempty set whose elements are called \emph{symbols}. An \emph{infinite word} $\sigma$ is an infinite sequence $\sigma = a_0 a_1 \dots \in \Sigma^\omega$ of  symbols with $a_i \in \Sigma$, $i \in \N$. For an infinite word $\sigma = \sigma_0 \sigma_1\dots \in \Sigma^\omega$ and $i \in \N$, let $\sigma(i) = \sigma_i$ denote the $i$-th symbol of $\sigma$ and $\sigma_\suf{i}$ the (infinite) suffix of $\sigma$ starting at position $i$, i.e., $\sigma_\suf{i} = \sigma_{i} \sigma_{i+1} \ldots \in \Sigma^\omega$. Notice that $\sigma_\suf{0}=\sigma$.

\subsection{Review of Linear Temporal Logic (LTL) and LTL model-checking}
\label{sec:LTL}
In this section, we describe the syntax and semantics of \emph{Linear Temporal Logic} (LTL) and recall the model-checking problem. This will form the backdrop against which rLTL will be introduced in the next section. The syntax of LTL is defined as follows.

\begin{definition}[LTL syntax]
\label{def:LTLsyntax}
	Let $\P$ be a nonempty, finite set of atomic propositions. The set of all \emph{LTL formulae} on $\P$, written $ \LTL(\P) $, is the smallest set satisfying:
	\begin{itemize}
		\item each $p \in \P$ is an LTL formula; and 
		\item if $\varphi$ and $\psi$ are LTL formulae, then so are $\neg \varphi$, $\varphi \land \psi$,  $\varphi \lor \psi$, $\varphi \implies \psi$, $\Next \varphi$, $\eventu \varphi$, $\always \varphi$, $\varphi \until \psi$, and $\varphi \release \psi$. 
	\end{itemize}
The closure of $\varphi$, denoted by $\cl(\varphi)$, is the set of its distinct subformulae, defined as:
\begin{itemize}
	\item $\cl(p) = \{ p \}$, if $p \in \mathcal{P}$ (atomic propositions);
	\item $\cl( \op ( \varphi ) )=\{ \op ( \varphi ) \} \cup \cl( \varphi )$,  if $\op \in \{ \neg, \Next, \eventu, \always \}$ (unary operators); and
	\item $\cl( \op(\varphi,\psi) )=\{ \op(\varphi,\psi)\} \cup \cl( \varphi ) \cup \cl( \psi )$,  if $\op \in \{ \land, \lor, \implies, \until, \release \}$ (binary operators).
\end{itemize}
The length of a formula $\varphi \in \LTL(\P)$, defined as $|\varphi| = \mathrm{card}(\cl(\varphi))$, is the number of distinct subformulae it contains. 
\end{definition}

For notational convenience, we have added syntactic sugar in the above definition by including the operators $\land$, $\implies$, $\always$, $\eventu$, $\release$ with their usual meaning as part of the syntax, although they can be defined using the $\neg$, $\lor$, and $\until$ operators. 

Usually, the semantics of LTL are defined in terms of a satisfiability relation between an LTL formula over the set of atomic propositions $\P$ and an infinite word over $\Sigma = 2^\P$. Having in mind the rLTL version of these notions, we provide a mathematically equivalent definition of the semantics as a mapping $W$ assigning an infinite word $\sigma \in \Sigma^\omega$ and an LTL formula $\varphi$ to the element $W(\sigma,\varphi) \in \B$. 

\begin{definition}[LTL Semantics]
\label{def:LTLsemantics}
The \emph{LTL semantics} is a mapping $W : \left ( 2^\P\right )^{\omega} \times \LTL(\P) \rightarrow \{0,1\}$, inductively defined as follows for $p \in \P$ and \mbox{$\varphi, \psi \in \LTL(\P)$}:
\begin{itemize}
	\item For atomic propositions: $W(\sigma, p) = \begin{cases} 0, \quad \text{ if $p \notin \sigma(0)$,} \\ 1, \quad \text{ if $p \in \sigma(0)$.} \end{cases}$
	\item For logical connectives:
	 \begin{align*}
		&W(\sigma, \neg \varphi) = 1 - W(\sigma, \varphi),  &W(\sigma, \varphi \land \psi) = \min{\{ W(\sigma, \varphi), W(\sigma, \psi) \}},	\\
		&W(\sigma, \varphi \lor \psi) = \max{\{ W(\sigma, \varphi), W(\sigma, \psi) \}}, &W(\sigma,\varphi\implies\psi) = \max{\{ W(\sigma, \neg\varphi), W(\sigma, \psi) \}}.
	\end{align*}
	\item For temporal operators:
	\begin{align*} 
	&W(\sigma, \Next \varphi) = W(\sigma_\suf{1}, \varphi),  \qquad W(\sigma, \eventu \varphi) = \sup_{i \geq 0}{W(\sigma_\suf{i}, \varphi)}, \qquad W(\sigma, \always \varphi) = \inf_{i \geq 0}{W(\sigma_\suf{i}, \varphi)}, \\
	&W(\sigma, \varphi \until \psi)  = \sup_{j \geq 0} \min \left \{  W(\sigma_\suf{j}, \psi),  \inf_{0 \leq i < j}{W(\sigma_\suf{i}, \varphi)} \right \}, \\
	&W(\sigma, \varphi \release \psi)  = \inf_{j\ge 0}\max\left\{W(\sigma_\suf{j},\psi),\sup_{0\le i<j}W(\sigma_\suf{i},\varphi)\right\}.
	\end{align*}
\end{itemize}
\vspace{-2mm}
The evaluation of the mapping $W$ on the LTL formula $\varphi$ and the infinite word $\sigma$, $W(\sigma,\varphi)$, is the valuation of $\varphi$ over $\sigma$. 
\end{definition}

Having introduced the semantics of LTL, we now recall the problem of \emph{LTL model-checking}~\cite{clarke1999mcbook, clarke2018mchandbook, kesten1998algverltl, lichtenstein1985concprogltl, pnueli2008temporaltesters, schnoebelen2002mccomplexity, vardi1986automtheorver}, which is essential in formal in verification. Given a \emph{model} of a system, the question is to decide whether or not all possible executions of the model satisfy a specification. Traditionally, these models are described by \emph{Kripke structures}~\cite[Section 2.2]{clarke2018mchandbook}, and the specifications are described by LTL formulae. 
\begin{definition}[Kripke structures]
	A Kripke structure over a set $\P$ of atomic propositions is a quadruple \mbox{$\K = \left( Q, Q_0, R, \ell \right)$}, where $Q$ is a finite set of states, $Q_0 \subseteq Q$ is a set of initial states, $R \subseteq Q \times Q$ is a set of transitions, and $\ell : Q \rightarrow 2^\P$ is the labeling function  that associates each state with a set of atomic propositions.
	
A \emph{path} in $\K$ is an infinite sequence of states $q_0 q_1 \dots \in Q^\omega$ such that $q_0 \in Q_0$ and $(q_i, q_{i+1}) \in R$ for all $i \in \N$. Any path induces a corresponding \emph{computation} $\ell(q_0) \ell(q_1) \dots \in \left( 2^\P \right)^\omega$ .
\end{definition}

In addition to Kripke structures, B\"uchi automata are another ingredient in the solution of the LTL model-checking problem. LTL formulae can be translated to \emph{B\"uchi Automata}  (BA)~\cite[Section 4.3]{baier2008principlesmcbook},~\cite[Section 4.2]{clarke2018mchandbook} and the set of words they recognize. 

\begin{definition}[B\"uchi Automaton]
A $ \ $ \emph{(non-deterministic) $ \ $ B\"uchi $ \ $ Automaton} $ \ $ (BA) $ \ $ is $ \ $ a $ \ $ quintuple \mbox{$\A = (Q, \Sigma, Q_0, \Delta, F)$} consisting of a nonempty, finite set $Q$ of states, a (finite) input alphabet $\Sigma$, a set of initial states $Q_0 \subseteq Q$, a (nondeterministic) transition relation $\Delta \subseteq Q \times \Sigma \times Q$, and a set  $F \subseteq Q$ defining the acceptance conditions.

The \emph{run} of a BA on a word $\sigma \in \Sigma^\omega$ (also called \emph{input}) is an infinite sequence of states $q_0 q_1 \ldots \in Q^\omega$ satisfying $q_0 \in Q_0$ and $(q_i, \sigma(i), q_{i+1}) \in \Delta$ for all $i \in \mathbb N$. A run $\rho$ is called \emph{accepting} if at least one of its infinitely often occurring states is in $F$.
The \emph{language} of a BA $\A$, denoted by $L(\A)$, is the set of all infinite words $\sigma \in \Sigma^\omega$ for which an accepting run of $\A$ exists.
\end{definition}

The translation of an LTL formula $\varphi$ to a BA $\A_{\varphi}$ is done in two steps: 1) the LTL formula $\varphi$ is translated to a \emph{Generalized B\"uchi Automaton} (GBA); and 2) the GBA is translated to a BA. 

\begin{definition}[Generalized B\"uchi Automaton]
A \emph{(non-deterministic) Generalized B\"uchi Automaton} (GBA) is a quintuple \mbox{$\G = (Q, \Sigma, Q_0, \Delta, \mathcal F)$} consisting of a nonempty, finite set $Q$ of states, a (finite) input alphabet $\Sigma$, a set of initial states $Q_0 \subseteq Q$, a (nondeterministic) transition relation $\Delta \subseteq Q \times \Sigma \times Q$, and a set $\mathcal{F} \subseteq 2^Q$ denoting the acceptance conditions.

The \emph{run} of a GBA is defined analogously to the run of a BA. The difference is that a run is \emph{accepting} if its set of infinitely often occurring states contains at least one state from each accepting set in $\mathcal{F}$. Note that there may be no accepting sets, in which case any infinite run trivially satisfies this property. 
\end{definition}

Transforming a GBA $\G$ into an equivalent BA $\A$ requires creating $ \mathrm{card}(\mathcal{F})$ many copies of $\G$.
The acceptance set of copy $\G_i$ is connected to the states of copy $\G_{\mathrm{mod}(i+1,   \mathrm{card}(\mathcal{F})) }$, $i=1, \dots, \mathrm{card}(\mathcal{F})$. 
Then, the accepting condition for $\A$ asks that any accepting state of the first copy is visited infinitely often. This implies that the accepting sets of each copy are visited infinitely often too. For more details see~\cite[Section 4.3.4]{baier2008principlesmcbook}.

\begin{proposition}[Section 5.2,~\cite{baier2008principlesmcbook}]
\label{prop:LTL2BA}
	Any LTL formula $\varphi$ can be translated to a GBA $\G_{\varphi}$ with at most $2^{|\varphi|}$ states and at most $|\varphi|$ accepting conditions. Moreover, the GBA $\G_{\varphi}$ can be translated to a BA $\A_{\varphi}$ with at most $|\varphi| \cdot 2^{|\varphi|}$ states. 
\end{proposition}
Therefore, the time complexity of translating an LTL formula $\varphi$ to a GBA $\A_{\varphi}$ is $\bigO \left( 2^{|\varphi|} \right)$, and that of translating an LTL formula $\varphi$ to a BA $\A_{\varphi}$ is $\bigO \left( |\varphi| \cdot 2^{|\varphi|} \right)$.

\begin{problem}[LTL model-checking]
\label{def:LTLmc}
Given a set of atomic propositions $\P$, a Kripke structure $\K$ and an LTL formula $\varphi$, do all the computations of $\mathcal{K}$ satisfy $\varphi$?
\end{problem}

The classical approach to solving Problem~\ref{def:LTLmc} has running time depending linearly on the size of the Kripke structure and exponentially on the length of the LTL formula.

\begin{corollary}[LTL model-checking]
\label{cor:LTLmc}
The standard procedure for model-checking an LTL formula $\varphi$ on a Kripke structure $\K$ is as follows~\cite[Section 5.2]{baier2008principlesmcbook},~\cite[Section 4]{clarke2018mchandbook}: \\
\hspace*{0.5em}1) Construct a BA $\A_{\K}$ such that $\A_{\K}$ accepts a computation $\pi \in \left( 2^\P \right)^\omega$if and only if $\pi$ is a computation of $\K$. The size of $\A_{\K}$ is linear in the size of $\K$, i.e., in its number of states denoted by $|\K|$. \\
\hspace*{0.5em}2) Construct a BA $\A_{\neg \varphi}$ recognizing the words satisfying the negation of $\varphi$, i.e., $\neg \varphi$. The size of $\A_{\neg \varphi}$ is exponential in $|\varphi|$, specifically $\bigO \left( |\varphi| \cdot 2^{|\varphi|} \right)$ by Proposition~\ref{prop:LTL2BA}. \\
\hspace*{0.5em}3) Compose $\A_{\K}$ with $\A_{\neg\varphi}$ to obtain $\A_{\K,\neg \varphi}$, which recognizes \emph{all the words of $L\left(\A_{\K}\right)$ that do not satisfy $\varphi$}, i.e., $L\left(\A_{\K,\neg\varphi}\right) = L\left(\A_{\K}\right) \cap L\left(\A_{\neg\varphi}\right)$. The size of $\A_{\K,\neg \varphi}$ is $\bigO \left( |\K| \cdot |\varphi| \cdot 2^{|\varphi|} \right)$.	\\
\hspace*{0.5em}4) Check the \emph{emptiness} of $L\left(\A_{\K,\neg\varphi}\right)$: if $\A_{\K,\neg \varphi}$ only recognizes the empty language, then $L(\A_{\K})$ satisfies $\varphi$. 

The time complexity of Step 4 in terms of the size of $\K$ and the length of the LTL formula $\varphi$ is:
\begin{equation}
\label{eq:LTLComplexity}
	\bigO \Big( |\K| \cdot |\varphi| \cdot 2^{|\varphi|} \Big).
\end{equation}
The above represents the classical, and tight, upper bound for the time complexity of LTL model-checking. 
\end{corollary}

In this work, we are concerned with solutions to the model-checking problem that employ a translation of formulae to GBAs as we described in the previous paragraph. For this reason, when discussing the complexity of the model-checking problem induced by different temporal logics, we focus on the size of the corresponding GBA. For example, for an LTL formula $\varphi$ the corresponding GBA has size $\bigO \left( 2^{|\varphi|} \right)$, and in Sections~\ref{sec:rLTLmodelchecking} and~\ref{sec:augmentedfragment} we will derive similar upper bounds for rLTL and an rLTL fragment respectively.

\section{The Syntax and Semantics of Robust Linear Temporal Logic}
\label{sec:rLTLsyntaxsem}

The main goal of \emph{robust Linear Temporal Logic} (rLTL) is to embed a notion of robustness into LTL. With this in mind, we crafted the syntax of rLTL so as to closely resemble that of LTL by using \emph{robust} versions of  LTL operators. In addition to defining the rLTL syntax, we also define the rLTL semantics in this section and justify the necessity of the many-valued semantics, i.e., the five rLTL truth values. We do so by first considering the $\rLTL_{\ralways,\reventu}(\P)$ fragment, i.e., the fragment of rLTL formulae that only allows the ``robust always'' $\ralways$ and ``robust eventually'' $\reventu$ temporal operators. This fragment was first introduced in~\cite{tabuada2016rLTL}, and in this section we extend those results from the fragment $\rLTL_{\ralways,\reventu}(\P)$ to full rLTL.

\subsection{rLTL Syntax}
\label{subsec:rLTLsyntax}
We begin by presenting the rLTL syntax.
\begin{definition}[$\mathrm{r}$LTL syntax]
\label{def:rLTLsyntax}
Let $\mathcal P$ be a nonempty, finite set of atomic propositions. The set of all \emph{rLTL formulae} on $\P$, written $ \rLTL(\P) $, is the smallest set satisfying:
\begin{itemize}
	\item each $p \in \P$ is an rLTL formula; and
	\item if $\varphi$ and $\psi$ are rLTL formulae, then so are $\neg \varphi$,  $\varphi \lor \psi$, $\varphi \land \psi$, $\varphi \rimplies \psi$, $\nextdot \varphi$, $\reventu \varphi$, $\ralways \varphi$, $\varphi \untildot \psi$, and $\varphi \releasedot \psi$.
\end{itemize}
The length of a formula $\varphi \in \rLTL(\P)$, denoted by $|\varphi|$, is the number of distinct subformulae it contains.
\end{definition}

Notice that in LTL, the conjunction and implication operators can be derived from negation and disjunction. This is no longer the case in rLTL since it has a many-valued semantics. On these grounds, we directly included conjunction and robust implication as part of the rLTL syntax in Definition~\ref{def:rLTLsyntax}. The same reason justifies the presence of the robust release operator $\releasedot$, which, in the case of LTL, can be derived from the until and negation operators as $\varphi\release \psi=\neg\left(\neg\varphi\until\psi\right)$.

\subsection{The $\rLTL_{\ralways,\reventu}(\P)$ fragment}
\label{subsec:rLTLboxdiam}
Consider, as a running example, the LTL formula $\always p$ with $p$ an atomic proposition. There is \emph{only one way} in which this formula can be satisfied, namely when $p$ holds at every time step. In contrast, there are \emph{several ways} in which this formula can be violated. Our goal is to find a semantics that distinguishes between these different ways. We aim for such distinction to be limited by what can be expressed in LTL so that we can easily leverage the wealth of existing results on LTL verification and synthesis. 

By intuitively investigating the different ways in which $\always p$ is violated over an infinite trace, we are able to distinguish the following four cases: (1) the worst possible violation occurs when $p$ fails to hold at every time step; (2) a slightly better scenario, which still violates $\always p$, is where $p$ holds for at most finitely many time instants; (3) better yet would be that $p$ holds at infinitely many instants, while still failing to hold at infinitely many instants; (4) finally, among all the possible ways in which $\always p$ can be violated, the most preferable case would be the one where $p$ fails to hold for at most finitely many time instants. Consequently, our robust semantics is designed to distinguish between satisfaction and these four possible different ways to violate $\always p$. However, as convincing as this argument might be, a question persists: in which sense can we regard these five alternatives as canonical?
\subsubsection{Why $5$ truth values?}
\label{subsubsec:why5boxdiam}
We answer this question by interpreting satisfaction of $\always p$ as a counting problem, and showing that there exists a partition of the infinite strings $\B^\omega = \{0,1\}^\omega$, which is induced by the $\LTL_{\always,\eventu}$ formulae. The semantics of the LTL $\always$ operator is given by:
\begin{align}
	W(\sigma,\always\varphi) = \inf_{i \in \N} W(\sigma_\suf{i},\varphi).
\end{align}
We make the following observation: the truth value of the LTL formula $\always \varphi$ on the infinite word $\sigma \in \Sigma^\omega$ is invariant under permutations of $W^{\omega}(\sigma,\varphi) = W(\sigma_\suf{0},\varphi) W(\sigma_\suf{1},\varphi) \cdots$. To make this observation precise, let $f : \N \rightarrow \N$ be a permutation of $\N$, i.e., a bijection. Then, we have:
\begin{align}
	\inf_{i \in \N} W(\sigma_\suf{i},\varphi) = \inf_{i \in \N} W(\sigma_\suf{f(i)},\varphi).
\end{align}
The above property shows that the $\always$ operator counts the number of zeros in the infinite string \mbox{$W^\omega(\sigma,\varphi) \in \{0,1\}^\omega$}. When this number is $0$, $W(\sigma, \always \varphi)$ is one (\true), and otherwise zero (\false). The position where zeros occur is not relevant, only their presence or their absence is. Towards characterizing how to count in $\LTL_{\always,\eventu}$, we first note that by successively applying permutations that swap position $i$ with position $i+1$ and leave all the remaining elements of $\N$ unaltered, we can transform any string $\rho \in \{0,1\}^\omega$ into one of the following forms $1^\omega, \ 0^k 1^\omega, \ (01)^\omega, \ 1^k 0^\omega, \ 0^\omega$, where $k \in \N$. Moreover, since $\LTL_{\always,\eventu}$ is stutter-free~\cite{peled1997stutterinvarTL}, it follows that the previous cases degenerate into the following five:
\begin{align}
\label{eq:5ways}
	1^\omega, \ 0 1^\omega, \ (01)^\omega, \ 1 0^\omega, \ \text{and } 0^\omega.
\end{align}
We interpret these as the ability to count if the number of zeros and ones is \emph{zero}, \emph{finite}, or \emph{infinite}. If we denote by $(z,o)$ the number of zeros and ones, with $z, o \in \{\text{zer}, \text{fin}, \text{inf}\}$, then we have:
\begin{itemize}
	\item $1^\omega$ corresponds to $(\text{zer},\text{inf})$.
	\item $01^\omega$ corresponds to $(\text{fin},\text{inf})$.
	\item $(01)^\omega$ corresponds to $(inf,\text{inf})$.
	\item $10^\omega$ corresponds to $(\text{inf},\text{fin})$.
	\item $0^\omega$ corresponds to $(\text{inf},\text{zer})$.
\end{itemize}
We can thus conclude the need for $5$ truth values to describe the $5$ different ways of counting zeros and ones. Concretizing this discussion, in particular \eqref{eq:5ways}, on the running example of the formula $\always p$, we obtain the following canonical forms that can be distinguished:
\begin{equation}
\label{Eq:CanonicalForms}
\{p\}^\omega, \quad 
\left(\{\neg p\} \{p\} \right)^+ \{p\}^\omega, \quad 
\left( \{ \neg p\} \{p\} \right)^\omega,\quad 
\left( \{ \neg p\} \{p\} \right)^+ \{ \neg p\}^\omega,\quad 
\text{and } \{ \neg p \}^\omega.
\end{equation}
It is no surprise that these are exactly the five cases discussed in the beginning of this subsection. 

The considerations in this section suggest the need for a semantics that is 5-valued rather than 2-valued so that we can distinguish between the aforementioned five cases. Therefore, we need to replace Boolean algebras by a different type of algebraic structure that can accommodate a 5-valued semantics. \emph{Da Costa algebras}, reviewed in the next section, are an example of such algebraic structures.
\subsubsection{da Costa Algebras}
\label{subsubsec:dacosta}
According to our running example $\always p$, the desired semantics should have one truth value corresponding to \true and four truth values corresponding to different shades of \false. It is instructive to think of truth values as the elements of $\B^4$, i.e., the four-fold Cartesian product of $\B$, that arise as the possible values of the $4$-tuple of LTL formulae:
\begin{equation}
\label{Eq:4TupleLTL}
(\always p,\eventu\always p, \always\eventu p, \eventu p).
\end{equation}

To ease notation, we denote such values interchangeably by $b = b_1 b_2 b_3 b_4$ and $b = (b_1, b_2, b_3, b_4)$ with $b_i\in \B$ for $i\in\{1,2,3,4\}$. The value 1111 then corresponds to \true since $\always p$ is satisfied. The most preferred violation of $\always p$ ($p$ fails to hold at only finitely many time instants) corresponds to $0111$, followed by $0011$ ($p$ holds at infinitely many instants and also fails to hold at infinitely many instants), $0001$ ($p$ holds at most at finitely many instants), and $0000$ ($p$ fails to hold at every time instant). Such preferences can be encoded in the linear order:
\begin{equation}
\label{eq:ordering}
	0000\prec 0001\prec 0011\prec 0111\prec 1111	,
\end{equation}
that renders the set:
\begin{align}
\label{eq:B5}
	\B_5 & = \left\{0000,0001,0011,0111,1111\right\}, 
\end{align}
a (bounded) distributive lattice with top element $\top=1111$ and bottom element $\bot=0000$. Formally, $\B_5$ is the subset of $\B^4$ consisting of the $4$-tuples $(b_1,b_2,b_3,b_4)\in \B^4$ satisfying the monotonicity property:
\begin{align}
\label{Monotonicity}
	i\leq_\N j\text{ implies }b_i\leq_\B b_j	,
\end{align}
where $i,j\in\{1,\ldots,4\}$, $\leq_\N$ is the natural order on the natural numbers and $\leq_\B$ is the natural order on the Boolean algebra $\B$. In $\B_5$, the meet $\sqcap$ can be interpreted as minimum and the join $\sqcup$ as maximum with respect to the order in~\eqref{eq:ordering}. We use $\sqcap$ and $\sqcup$ when discussing  lattices in general and use $\min$ and $\max$ for the specific lattice $\B_5$ or the Boolean algebra $\B$.

The first choice to be made in using the lattice $(\B_5,\min,\max)$ to define the semantics of $\rLTL_{\ralways,\reventu}(\P)$ is the choice of an operation on $\B_5$ modeling conjunction. It is well know that all the desirable properties of a many-valued conjunction are summarized by the notion of triangular-norm, see~\cite{hajeck1998fuzzylogic,novak1999fuzzylogic}. One can compare two triangular-norms $s$ and $t$ using the partial order defined by declaring $s\le t$ when $s(a,b)\le t(a,b)$ for all $a,b\in \B_5$. According to this order, the triangular-norm $\min$ is maximal among all triangular-norms (i.e., we have $t(a,b)\le \min\{a,b\}$ for every $a,b\in \B_5$ and every triangular-norm $t$). This shows that if we choose any triangular-norm $t$ different from $\min$, there exist elements $a,b\in \B_5$ for which we have $t(a,b)<\min\{a,b\}$. Hence, any choice different from $\min$ would result in situations where the value of a conjunction is \emph{smaller} than the value of the conjuncts, which is not reasonable when interpreting the value of the conjuncts as different shades of \false. To illustrate this point, consider the formula $\always p\land \always q$ and the word $\sigma=\emptyset \{p,q\}^\omega$. As introduced above, the value of $\always p$ on $\sigma$ corresponds to $0111$ and the value of $\always q$ on $\sigma$ corresponds to $0111$ since on both cases we have the most preferred violation of the formulae. Therefore, the value of $\always p\land \always q$ on $\sigma$ should also be $0111$ since the formula $\always p\land \always q$ is only violated a finite number of times. It thus seems natural\footnote{Note that there are situations where it is convenient to model conjunction differently. 
In the work of Bloem et~al.~\cite{bloem2010robustnessliveness}, the specific way in which robustness is modeled requires distinguishing between the number of conjuncts that are satisfied in the assumption $\land_{i\in I} \varphi_i$. This cannot be accomplished if conjunction is modeled by $\min$, and a different triangular-norm would have to be used for this purpose. Note that both \L{}ukasiewicz's conjunction as well as Goguen's conjunction 
have the property that their value decreases as the number of conjuncts that are true decreases.} to model conjunction in $\B_5$ by $\min$ and, for similar reasons, to model disjunction in $\B_5$ by $\max$. 

As in intuitionistic logic, our implication is defined as the residue of $\sqcap$\footnote{This is also done in context of residuated lattices that is more general than the Heyting algebras used in intuitionistic logic. Recall that a residuated lattice is a lattice $(A,\sqcap,\sqcup)$, satisfying the same additional conditions, and equipped with a commutative monoid $(A,\otimes,\textbf{1})$ satisfying additional compatibility conditions. Since we chose the lattice meet $\sqcap$ to represent conjunction, we have a residuated lattice where $\otimes=\sqcap$ and $\textbf{1}=\top$.}. In other words, we define the implication $a\rightarrow b$ by requiring that $c \preceq a \rightarrow b$ if and only if $c \sqcap a \preceq b$ for every $c\in \B_5$. This leads to:
\begin{align*}
a\rightarrow b= \begin{cases} 1111 & \text{if $a\preceq b$; and} \\ b & \text{otherwise.} \end{cases} 
\end{align*}

However, we now \emph{diverge} from intuitionistic logic (and most many-valued logics) where negation of $a$ is defined by $a\rightarrow 0000$. Such negation is not compatible with the interpretation that all the elements of $\B_5$, except for $1111$, represent (different shades of) \false and thus their negation should have the truth value $1111$. To make this point clear, we present in Table~\ref{Neg} the intuitionistic negation in $\B_5$ and the desired negation compatible with the  interpretation of the truth values in $\B_5$.  

\begin{table}[h]
	\centering
	\caption{Desired negation vs.\ intuitionistic negation in $\B_5$.}\label{Neg}
	\vspace{-2.5mm}
	\begin{tabular}{c@{\hskip 3em}cc}	
		\toprule
		Value & Desired negation & Intuitionistic negation\\
		\midrule
		1111 & 0000 & 0000\\
		0111 & 1111 & 0000\\
		0011 & 1111 & 0000\\
		0001 & 1111 & 0000\\
		0000 & 1111 & 1111\\
		\bottomrule
	\end{tabular}
\end{table}

What is then the algebraic structure on $\B_5$ that supports the desired negation, dual to the intuitionistic negation? This very same problem was investigated in~\cite{priest09dualintuitlogic}, and the answer is \emph{da~Costa} algebras. 

\begin{definition}[da~Costa algebra]
\label{Def:daCostaAlgebra}
A \emph{da Costa} algebra is a sextuple $(A,\sqcap,\sqcup,\preceq,\rightarrow,\overline{\,\cdot\,})$ where:
\begin{itemize} 
	\item $(A,\sqcap,\sqcup,\preceq)$ is a distributive lattice where $\preceq$ is the ordering relation derived from $\sqcap$ and $\sqcup$,
	\item $\rightarrow$ is the residual of $\sqcap$ (i.e., $a\preceq b\rightarrow c$ if and only if $a\sqcap b\preceq c$ for every $a,b,c\in A$),
	\item $a \preceq b\sqcup\overline{b}$ for every $a,b\in A$, and
	\item $\overline{a} \preceq b$ whenever $c\sqcup\overline{c}\preceq a\sqcup b$ for every $a,b,c\in A$.
\end{itemize}
\end{definition}

In a da~Costa algebra, one can define the top element $\top$ to be $\top=a\sqcup\overline{a}$ for an arbitrary $a\in A$; note that $\top$ is unique and independent of the choice of $a$. Hence, the third requirement in Definition~\ref{Def:daCostaAlgebra} amounts to the definition of top element, while the fourth requirement can be simplified to $\overline{a} \preceq b$ whenever $\top\preceq a\sqcup b$. We can easily verify that $\B_5$ is a da Costa algebra if we use the desired negation defined in Table~\ref{Neg}. 

It should be mentioned that working with a $5$-valued semantics has its price. The law of non-contradiction fails in $\B_5$ (i.e., $a \sqcap \overline{a}$ may not equal $\bot=0000$ as evidenced by taking $a=0111$). However, since $a \sqcap \overline{a}\prec 1111$, a weak form of non-contradiction still holds as $a \sqcap \overline{a}$ is to be interpreted as a shade of \false but not necessarily as the least preferred way of violating $a \sqcap \overline{a}$, which corresponds to $\bot$. Contrary to intuitionistic logic, the law of excluded middle is valid (i.e., $a \sqcup \overline{a}=\top=1111$). Finally, $a=0111$ shows that $\overline{\overline{a}}\ne a$, although it is still true that $\overline{\overline{a}}\rightarrow a$. Interestingly, we can think of the double negation:
\begin{align*}
	 \overline{\overline{a}} = \begin{cases} 1111, & \text{if $a=1111$,} \\ 0000, & \text{otherwise,} \end{cases}
\end{align*}
as quantization in the sense that \true is mapped to \true and all the shades of \false are mapped to \false. Hence, double negation quantizes the five different truth values into two truth values (\true and \false) in a manner that is compatible with our interpretation of truth values.

\subsubsection{Semantics of $\rLTL_{\ralways,\reventu}(\P)$ on da~Costa Algebras}
\label{subsubsec:semanticsboxdiam}
The semantics of $\rLTL_{\ralways,\reventu}(\P)$ is given by a mapping $V$, called valuation as in the case of LTL, that maps an infinite word $\sigma\in \Sigma^\omega$ and an $\rLTL_{\ralways,\reventu}(\P)$ formula $\varphi$ to an element of $\B_5$. In defining $V$, we judiciously use the algebraic operations of the da Costa algebra $\B_5$ to give meaning to the logical connectives in the syntax of $\rLTL_{\ralways,\reventu}(\P)$. In the following, let $\Sigma=2^\P$, where $\P$ is a finite set of atomic propositions. 

The semantics of rLTL is a mapping $V : \left(2^\P\right)^\omega \times \rLTL(\P) \rightarrow \B_5$. On atomic propositions $p\in \P$, $V$ is defined by:
\begin{align}
V(\sigma, p) = \begin{cases}
0000, & \text{if $p \notin \sigma(0)$; and} \\
1111, & \text{if $p \in \sigma(0)$.}
\end{cases}
\end{align}
Hence, atomic propositions are interpreted classically, i.e., only two truth values are used. Since we are using a $5$-valued semantics, we provide a separate definition for all the four logical connectives:
\begin{align}
V(\sigma,\neg\varphi) &= \overline{V(\sigma,\varphi)}, \\
V(\sigma,\varphi\land\psi) &= V(\sigma,\varphi)\sqcap V(\sigma,\psi), \\
V(\sigma,\varphi\lor\psi) &= V(\sigma,\varphi)\sqcup V(\sigma,\psi), \\
V(\sigma,\varphi\rimplies\psi) &= V(\sigma,\varphi) \rightarrow V(\sigma,\psi)	\label{eq:semanticsrimplies}.
\end{align}
Note how the semantics mirrors the algebraic structure of da Costa algebras. This is no accident since valuations are typically algebra homomorphisms.

Unfortunately, da Costa algebras are not equipped\footnote{One could consider developing a notion of da Costa algebras with operators in the spirit of Boolean algebras with operators~\cite{jonsson1951booleanalgebrasoperators}. We leave such investigation for future work.} with operations corresponding to $\ralways$ and $\reventu$, the robust versions of $\always$ and $\eventu$, respectively. Therefore, we resort to the counting interpretation in Section~\ref{subsubsec:why5boxdiam} to motivate the semantics of $\ralways$. Formally, the semantics of $\ralways$ is given by:
\begin{align}
\label{eq:ralwayssemantics}
V(\sigma,\ralways\varphi)=\left(\inf_{i\ge 0}V_1(\sigma_\suf{i},\varphi),~ \sup_{j\ge 0}\inf_{i\ge j}V_2(\sigma_\suf{i},\varphi),~ \inf_{j\ge 0}\sup_{i\ge j}V_3(\sigma_\suf{i},\varphi),~ \sup_{i\ge 0}V_4(\sigma_\suf{i},\varphi)\right),
\end{align}
where $V_k(\sigma,\varphi)=\pi_k\circ V(\sigma,\varphi)$ for $k\in \{1,2,3,4\}$ and $\pi_k:\B_5\to \B$ are the projections on the $k$-th element of a truth value defined by:
\begin{equation}
\label{Eq:Projection} 
\pi_k(a_1,a_2,a_3,a_4)=a_k.
\end{equation}

To illustrate the semantics of $\ralways$, let us consider the simple case where $\varphi$ is just an atomic proposition $p$. This means that one can express $V(\sigma,\ralways p)$ in terms of the LTL valuation $W$ by:
\begin{align}
\label{Decompose}
V(\sigma,\ralways p)=\left(W(\sigma,\always p),W(\sigma,\eventu\always p),W(\sigma,\always\eventu p),W(\sigma,\eventu p)\right).
\end{align}
In other words, $V_1(\sigma,\ralways p)$ corresponds to the LTL truth value of $\always p$, $V_2(\sigma, \ralways p)$ corresponds to the LTL truth value of $\eventu\always p$, $V_3(\sigma, \ralways p)$ corresponds to the LTL truth value of $\always\eventu p$, and $V_4(\sigma, \ralways p)$ corresponds to the LTL truth value of $\eventu p$. Equation~\eqref{Decompose} connects the semantics of $\ralways$ to the counting problems described in Section~\ref{subsubsec:why5boxdiam} and to the $4$-tuple of LTL formulae in~\eqref{Eq:4TupleLTL}. 

The last operator is $\reventu$, whose semantics is given by:
\begin{align}
V(\sigma,\reventu\varphi)=\left(\sup_{i\ge 0}V_1(\sigma_\suf{i},\varphi),~ \sup_{i\ge 0}V_2(\sigma_\suf{i},\varphi),~ \sup_{i\ge 0}V_3(\sigma_\suf{i},\varphi),~ \sup_{i\ge 0}V_4(\sigma_\suf{i},\varphi)\right).
\end{align}
According to the counting problems used in Section~\ref{subsubsec:why5boxdiam} to motivate the proposed semantics, there is only one way in which the LTL formula $\eventu p$, for an atomic proposition $p$, can be violated. Hence, $V(\sigma,\reventu p)$ is one of only two possible truth values: 1111 or 0000. We further note that $\reventu$ is not dual to $\ralways$, as expected in a many-valued logic where the law of double negation fails.

Having defined the semantics of $\rLTL_{\ralways,\reventu}(\P)$, let us now see if the formula $\ralways p\rimplies \ralways q$, where $\ralways p$ is an environment assumption and $\ralways q$ is a system guarantee with $p,q\in\P$, lives to the expectations set in the introduction and to the intuition provided in Section~\ref{subsubsec:why5boxdiam}. \\
\hspace*{0.5em}1) According to~\eqref{Decompose}, if $\always p$ holds, then $\ralways p$ evaluates to $1111$ and the implication $\ralways p\rimplies \ralways q$ is \true, i.e., the value of $\ralways p\rimplies \ralways q$ is 1111, if $\ralways q$ evaluates to $1111$, that is, if $\always q$ holds. Therefore, the desired behavior of $\always p\implies \always q$, when the environment assumptions hold, is retained.	\\
\hspace*{0.5em}2) Consider now the case where $\always p$ fails but the weaker assumption $\eventu\always p$ holds. In this case $\ralways p$ evaluates to $0111$ and the implication $\ralways p \rimplies \ralways q$ is \true if $\ralways p$ evaluates to $0111$ or higher. This means that $\eventu\always q$ needs to hold. \\
\hspace*{0.5em}3) A similar argument shows that we can also conclude the following consequences whenever $\ralways p\rimplies \ralways q$ evaluates to $1111$: $\always\eventu q$ follows whenever the environment satisfies $\always\eventu p$ and $\eventu q$ follows whenever the environment satisfies $\eventu p$.
We thus conclude that the semantics of $\ralways p\rimplies \ralways q$ captures the desired robustness property by which a weakening of the assumption $\ralways p$ leads to a weakening of the guarantee $\ralways q$. 
\subsection{Full $\rLTL(\P)$}
\label{subsec:fullrLTL}
In this section, we present the semantics of full $\rLTL(\P)$ by providing the semantics for three additional operators: the ``robust release'', denoted by $\releasedot$, the ``robust until'', denoted by $\untildot$, and the ``next'', denoted by $\nextdot$, operators. 

\subsubsection{Why $5$ truth values?}
\label{subsubsec:why5fullrLTL}
We revisit this question to better motivate the full rLTL semantics. According to the safety-progress classification of temporal properties, eloquently put forward in~\cite{chang1993safetyprogress, manna1987hierarchy}, $\always p$ defines a safety property. It can be expressed as $A(L)$ with $L$ being the regular language $\Sigma^* \{p\}$ and $A$ the operator generating all the infinite words in $\left(2^\P\right)^\omega$ with the property that all its finite prefixes belong to $L$. In addition to $A$, we can find in~\cite{chang1993safetyprogress, manna1987hierarchy} the operators $E$, $R$, and $P$ defining guarantee, response, and persistence properties, respectively. The language $E(L)$ consists of all the infinite words that contain at least one prefix in $L$, the language $R(L)$ consists of all the infinite words that contain infinitely many prefixes in $L$, and the language $P(L)$ consists of all the infinite words such that all but finitely many prefixes belong to $L$. Using these operators we can reformulate the semantics of $\ralways p$ from \eqref{eq:ralwayssemantics} as:
\begin{equation}
\label{SemanticsViaAPRE}
V(\sigma,\ralways p) = \begin{cases}
	1111 & \text{if $\sigma \in A(L)$;} \\
	0111 & \text{if $\sigma \in P(L)\setminus A(L)$;} \\
	0011 & \text{if $\sigma \in R(L) \setminus \left( A(L)\cup P(L) \right)$;} \\
	0001 & \text{if $\sigma \in E(L) \setminus \left( A(L)\cup P(L) \cup R(L) \right)$; and} \\
	0000 & \text{if $\sigma \notin E(L)$.}
\end{cases}
\end{equation}

We thus obtain a different justification for the five different truth values used in rLTL and why the five different cases in~\eqref{Eq:CanonicalForms} can be seen as canonical. Moreover, we can build on this perspective to define the semantics of the robust release and the robust until.
\subsubsection{Semantics of full $\rLTL(\P)$}
\label{subsubsec:semanticsfullrLTL}
We are now ready to define the semantics for the ``robust release'' $\releasedot$, and ``robust until'' $\untildot$ operators. Equality \eqref{SemanticsViaAPRE} suggests how we can define the $5$-valued semantics for the release operator. Recall that the LTL formula $p\,\mathcal{R}\,q$, for atomic propositions $p$ and $q$, defines a safety property, and that its semantics is given by:
\begin{equation}
\label{SemanticsRelease}
	W(\sigma, p \release q ) = \inf_{j\ge 0}\max\left\{V_1(\sigma_\suf{j},q),\sup_{0\le i<j}V_1(\sigma_\suf{i},p)\right\}.
\end{equation}
We can interpret $\max \left\{ V_1(\sigma_\suf{j},q),\sup_{0\le i<j}V_1(\sigma_\suf{i},p) \right\}$ above 
as the definition of the regular language $L=\Sigma^* \{q\}+\Sigma^*\{p\} \Sigma^+$ and $\inf_{j\ge 0}$ 
as the requirement that every prefix of a string satisfying $p\,\mathcal{R}\,q$ belongs to $L$, i.e., as the definition of the operator $A$. Therefore, the $5$-valued semantics can be obtained by successively enlarging the language $A(L)$ through the replacement of the operator $A$, formalized by $\inf$ in Equation~\eqref{SemanticsRelease}, by the operators $P$ formalized by $\sup \inf$, $R$ formalized by $\inf\sup$, and $E$ for formalized by $\sup$.  This observation leads to the semantics:
\begin{align}
	V(\sigma, \varphi \releasedot \psi ) = \left(V_1(\sigma,\varphi \releasedot \psi ),V_2(\sigma,\varphi \releasedot \psi ),V_3(\sigma, \varphi \releasedot \psi ),V_4(\sigma, \varphi \releasedot \psi )\right),
\end{align}
where:
\begin{align*}
	&V_1(\sigma,\varphi \releasedot \psi ) = \inf_{j\ge 0}\max \left\{V_1(\sigma_\suf{j},\psi),\sup_{0\le i<j} V_1(\sigma_\suf{i},\varphi)\right\}, \\	
	&V_2(\sigma,\varphi \releasedot \psi ) = \sup_{k\ge 0}\inf_{j\ge k}\max \left\{V_2(\sigma_\suf{j},\psi),\sup_{0\le i<j} V_2(\sigma_\suf{i},\varphi)\right\},	\\
	&V_3(\sigma, \varphi \releasedot \psi ) = \inf_{k\ge 0}\sup_{j\ge k}\max \left\{V_3(\sigma_\suf{j},\psi),\sup_{0\le i<j} V_3(\sigma_\suf{i},\varphi)\right\}, \\
	&V_4(\sigma,\varphi \releasedot \psi ) = \sup_{j\ge 0}\max \left\{V_4(\sigma_\suf{j},\psi),\sup_{0\le i<j} V_4(\sigma_\suf{i},\varphi)\right\}.
\end{align*}

Similarly to LTL, $\ralways\psi = \false \releasedot \psi$ holds, thereby showing that the semantics for the $\releasedot$ operator is compatible with the semantics of the $\ralways$ operator. We glean further intuition behind the definition of $\releasedot$ by considering the special case where $\varphi$ is given by $p$ and $\psi$ is given by $q$, for two atomic propositions $p, q \in \mathcal P$. Expressing $V(\sigma, p \releasedot q)$ in terms of an LTL valuation $W$, we obtain:
\begin{align*}
	V(\sigma, p \releasedot q)=\left(W(\sigma, p \release q),~ W(\sigma,\eventu\always q \lor \eventu p),~ W(\sigma,\always\eventu q\lor \eventu p),~ W(\sigma,\eventu q\lor \eventu p)\right).
\end{align*}

As long as $p$ occurs, the value of $p \releasedot q$ is at least $0111$. It could be argued that the semantics of $p \releasedot q$ should also count the number of occurrences of $q$ preceding the first occurrence of $p$. As we detail in Section~\ref{subsec:semanticsexamples}, this can be expressed in rLTL by making use of the proposed semantics.

In LTL, the $\until$ operator is dual to the $\release$ operator, but such relationship does not extend to rLTL in virtue of how negation was defined. Hence, the semantics of the $\untildot$ operator has to be introduced separately. Analogously to above, we interpret the LTL semantics of $p \untildot q$, given by:
\begin{equation}
\label{SemanticsUntil}
W(\sigma, p \until q ) = \sup_{j\ge 0}\min\left\{V_1(\sigma_\suf{j},q),\inf_{0\le i<j}V_1(\sigma_\suf{i},p)\right\},
\end{equation}
as defining the language $E\left(\{p\}^*\{q\}\right)$. In the hierarchy of the operators $E$, $R$, $P$, and $A$, defined by the inclusions $A(L)\subset P(L)\subset R(L)\subset E(L)$ for any regular language $L$, the language $E\left(\{p\}^*\{q\}\right)$ cannot be enlarged as it sits at the top of the hierarchy. Therefore, the semantics of the $\untildot$ operator is given by:
\begin{align}
\label{eq:untilsemantics}
	V(\sigma, \varphi \untildot \psi ) = \left(V_1(\sigma,\varphi \untildot \psi ),V_2(\sigma,\varphi \untildot \psi ),V_3(\sigma, \varphi \untildot \psi ),V_4(\sigma, \varphi \untildot \psi )\right),
\end{align}
where:
\begin{align*}
	V_k(\sigma,\varphi \untildot \psi ) = \sup_{j\ge 0}\min\left\{V_k(\sigma_\suf{j},\psi),\inf_{0\le i<j}V_k(\sigma_\suf{i},\varphi)\right\}, \text{ for each } k\in\{1,2,3,4\}.
\end{align*}

We obtain, by definition, that the semantics of the $\untildot$ operator is compatible with the semantics of the $\reventu$ operator in the sense that $\reventu \psi = \true \untildot \psi$.

Finally, we define the robust semantics of next as a direct generalization of the LTL semantics from $\B$ to $\B_5$:
\begin{align}
	V(\sigma,\nextdot\varphi)=V(\sigma_\suf{1},\varphi).
\end{align}
\subsection{Review of rLTL semantics}
\label{subsec:rLTLsemanticsreview}
To recap the previous discussion, we compactly present the formal rLTL semantics below.
\begin{definition}[$\mathrm{r}$LTL Semantics]
\label{def:rLTLsemantics}
The \emph{rLTL semantics} is a mapping $V : \left ( 2^\P\right )^{\omega} \times \rLTL(\P) \rightarrow \B_4$,  inductively defined as follows for $p \in \P$ and \mbox{$\varphi, \psi \in \rLTL(\P)$}:
\begin{itemize}
	\item For atomic propositions: $V(\sigma, p) = 	\begin{cases}
										0000, \quad \text{if $p \notin \sigma(0)$,} \\
										1111, \quad \text{if $p \in \sigma(0)$.}
										\end{cases}$
	\item For logical connectives: \quad
	$\begin{aligned}
		&V(\sigma,\neg\varphi) = \overline{V(\sigma,\varphi)},  &V(\sigma,\varphi\land\psi) = V(\sigma,\varphi)\sqcap V(\sigma,\psi),	\\
		&V(\sigma,\varphi\lor\psi) = V(\sigma,\varphi)\sqcup V(\sigma,\psi), &V(\sigma,\varphi\rimplies\psi) = V(\sigma,\varphi) \rightarrow V(\sigma,\psi).
	\end{aligned}$
	\item For temporal operators:
		\begin{align*}
		&V(\sigma,\nextdot\varphi) = V(\sigma_\suf{1},\varphi),	\\
		&V(\sigma,\reventu\varphi) = \left(\sup_{i\ge 0}V_1(\sigma_\suf{i},\varphi),~ \sup_{i\ge 0}V_2(\sigma_\suf{i},\varphi),~ \sup_{i\ge 0}V_3(\sigma_\suf{i},\varphi),~ \sup_{i\ge 0}V_4(\sigma_\suf{i},\varphi)\right),	\\
		&V(\sigma,\ralways\varphi) = \left(\inf_{i\ge 0}V_1(\sigma_\suf{i},\varphi),~ \sup_{j\ge 0}\inf_{i\ge j}V_2(\sigma_\suf{i},\varphi),~ \inf_{j\ge 0}\sup_{i\ge j}V_3(\sigma_\suf{i},\varphi),~ \sup_{i\ge 0}V_4(\sigma_\suf{i},\varphi)\right),\\
		 &V(\sigma, \varphi \untildot \psi ) = \left(V_1(\sigma,\varphi \untildot \psi ),V_2(\sigma,\varphi \untildot \psi ),V_3(\sigma, \varphi \untildot \psi ),V_4(\sigma, \varphi \untildot \psi )\right),	\\
		& \hspace{-5em} \text{where: } \ 
		V_k(\sigma,\varphi \untildot \psi ) = \sup_{j\ge 0}\min\left\{V_k(\sigma_\suf{j},\psi),\inf_{0\le i<j}V_k(\sigma_\suf{i},\varphi)\right\}, \text{ for each } k\in\{1,2,3,4\}.\\
		&V(\sigma, \varphi \releasedot \psi ) = \left(V_1(\sigma,\varphi \releasedot \psi ),V_2(\sigma,\varphi \releasedot \psi ),V_3(\sigma, \varphi \releasedot \psi ),V_4(\sigma, \varphi \releasedot \psi )\right), \\
		&\hspace{-5em} \text{where: }
		\end{align*}
		\vspace{-7.5mm}
		\begin{align*}
		&V_1(\sigma,\varphi \releasedot \psi ) = \inf_{j\ge 0}\max \left\{V_1(\sigma_\suf{j},\psi),\sup_{0\le i<j} V_1(\sigma_\suf{i},\varphi)\right\}, \\	
		&V_2(\sigma,\varphi \releasedot \psi ) = \sup_{k\ge 0}\inf_{j\ge k}\max \left\{V_2(\sigma_\suf{j},\psi),\sup_{0\le i<j} V_2(\sigma_\suf{i},\varphi)\right\},	\\
		&V_3(\sigma, \varphi \releasedot \psi ) = \inf_{k\ge 0}\sup_{j\ge k}\max \left\{V_3(\sigma_\suf{j},\psi),\sup_{0\le i<j} V_3(\sigma_\suf{i},\varphi)\right\}, \\
		&V_4(\sigma,\varphi \releasedot \psi ) = \sup_{j\ge 0}\max \left\{V_4(\sigma_\suf{j},\psi),\sup_{0\le i<j} V_4(\sigma_\suf{i},\varphi)\right\}.
		\end{align*}
\end{itemize}

\end{definition}

\subsection{Examples}
\label{subsec:semanticsexamples}
In this section, we showcase the applicability of the proposed semantics with a number of examples.
\subsubsection{The usefulness of implications that are \textbf{not true}}
We argued in the previous section that rLTL captures the intended robustness properties for the specification $\ralways p\rimplies \ralways q$ whenever this formula evaluates to $1111$. But does the formula $\ralways p\rimplies \ralways q$ still provide useful information when its value is lower than $1111$? It follows from the semantics of implication that $V(\sigma, \ralways p\rimplies \ralways q)=b$, for $b\prec 1111$, occurs when $V(\sigma,\ralways q)=b$, i.e., whenever a value of $b$ can be guaranteed despite $b$ being smaller than $V(\sigma,\ralways p)$. The value $V(\sigma, \ralways p\rimplies \ralways q)$ thus describes which weakened guarantee follows from the environment assumption whenever the intended system guarantee does not. This can be seen as another measure of robustness: despite $\ralways q$ not following from $\ralways p$, the behavior of the system is not arbitrary, a value of $b$ is still guaranteed. The usefulness that the different shades of false provide is further discussed in Section~\ref{sec:casestudies} in the context of a concrete case study.

\subsubsection{GR(1) in rLTL}
The GR(1) fragment of LTL is becoming increasingly popular for striking an interesting balance between its expressiveness and the complexity of the corresponding synthesis problem~\cite{bloem2012GR1}.  A GR(1) formula is an LTL$(\always,\eventu)$ formula of the form:
\begin{equation}
\label{GR(1)}
\bigwedge_{i\in I}\always\eventu p_i\implies \bigwedge_{j\in J}\always\eventu q_j ,
\end{equation}
where $p_i$ and $q_j$ are atomic propositions and $I, J$ are finite sets. We obtain the rLTL version of~\eqref{GR(1)} simply by dotting the boxes and the diamonds: 
\begin{equation}
\label{rGR(1)}
\bigwedge_{i\in I}\ralways\reventu p_i\rimplies \bigwedge_{j\in J}\ralways\reventu q_j.
\end{equation}
Any valuation $V$ for $\ralways\reventu p_i$ can be expressed in terms of a  valuation $W$ for LTL as:
\begin{align*}
V(\sigma, \ralways\reventu p_i) & =\left(W(\sigma, \always\eventu p_i),~ W(\sigma, \always\eventu p_i),~ W(\sigma, \always\eventu p_i),~ W(\sigma, \eventu p_i)\right).
\end{align*}
Therefore, $V(\sigma, \ralways\reventu p_i)$ can only assume three different values: $1111$ when $\always\eventu p_i$ holds, $0001$ when $\always\eventu p_i$ fails to hold but $\eventu p_i$ does hold, and $0000$ when $\eventu p_i$ fails to hold. Based on this observation, and assuming that~\eqref{rGR(1)} evaluates to $1111$, we conclude that $\bigwedge_{j\in J}\always\eventu q_j$ holds whenever $\bigwedge_{i\in I}\always\eventu p_i$ does, as required by~\eqref{GR(1)}. In contrast to \eqref{GR(1)}, however, the weakened system guarantee $\bigwedge_{j\in J}\eventu q_j$ holds whenever the weaker environment assumption $\bigwedge_{i\in I}\eventu p_i$ does.	

\subsubsection{Non-counting formulae}
All the preceding examples were counting formulae. We now consider the simple non-counting formula $\always(p\implies \eventu q)$, which requires each occurrence of $p$ to be followed by an occurrence of $q$. The word $\sigma_1 = \{p\} \{q\} \emptyset^\omega$ clearly satisfies this formula although its permutation $\sigma_2 = \{q\} \{p\} \emptyset^\omega$ does not. In addition to being a non-counting formula, $\always(p\implies \eventu q)$ is one of the most popular examples of an LTL formula used in the literature known as the ``request-response'' property, and, for this reason, constitutes a litmus test to rLTL. Actually, such formula is part of the reactive specification patterns identified in~\cite{dwyer1999patterns}, which are part of our case studies in Section~\ref{sec:casestudies}. The semantics of the dotted version of $\always (p\implies \eventu q)$ can be expressed using an LTL valuation $W$ as:
\begin{align*}
	V(\sigma,\ralways(p\rimplies \reventu q)) =& \\
	&\hspace*{-5em}
	\left(W(\sigma,\always(p\implies \eventu q)), W(\sigma,\always\eventu p\implies \always\eventu q), W(\sigma,\eventu\always p\implies \always\eventu q), W(\sigma,\always p\implies \eventu q)\right).
\end{align*}
It is interesting to observe how the semantics of $\varphi=\ralways(p\rimplies\reventu q)$ recovers: strong fairness, also known as compassion, when the value of $\varphi$ is 0111; weak fairness, also known as justice, when the value of $\varphi$ is 0011; and the even weaker notion of fairness represented by the LTL formula $\always p\implies \eventu q$, when the value of $\varphi$ is 0001. The fact that all these different and well known notions of fairness naturally appear in the proposed semantics is another strong indication of rLTL's naturalness and usefulness.	\\

\subsubsection{Counting with ``robust release''} As we discussed before, the semantics of $\varphi \releasedot \psi$ does not count how many times $\psi$ holds before the first occurrence of $\varphi$. This property, however, is captured by the rLTL formula:
\begin{equation}
\label{Eq:ReleasedotWithCounting}
	\left(\varphi\releasedot \psi\right)\land \left(\neg\varphi\untildot \psi\right).
\end{equation} 
To see why, we assume $\varphi=p$ and $\psi=q$ for atomic propositions $p$ and $q$. Then, express the semantics of the rLTL formula~\eqref{Eq:ReleasedotWithCounting} in terms of an LTL valuation $W$ as:
\begin{align*}
	V(\sigma,\left(p\releasedot q\right)\land \left(\neg p\untildot q\right)) & =\left(W(\sigma, p\release q),W(\neg p\until q),W(\neg p\until q),W(\neg p\until q)\right). 
\end{align*}
Note how we can now distinguish between three cases: (1) $p\release q$ holds, corresponding to value $1111$; (2) $q$ holds at least once before being released by $p$, corresponding to value $0111$; and (3) $q$ does not hold before being released by $p$, corresponding to value $0000$.	\\

\subsubsection{Non-decomposition of ``robust until''} The previous example showed how the LTL equivalence between $\varphi \release \psi$ and $\left( \varphi \release \psi\right)\land\left(\neg \varphi \until \psi\right)$ is not valid in rLTL. Another LTL equivalence that is not valid in rLTL is the decomposition of the until operator into its liveness and safety parts, that is the equivalence between $\varphi \until \psi$ and $\eventu\psi\land (\psi \release (\psi\lor\varphi))$.  The rLTL formula $\reventu\psi\land (\psi \releasedot (\psi\lor\varphi))$ expresses a weaker requirement than $\varphi \untildot \psi$ that is also useful to express robustness. When $\varphi$ and $\psi$ are the atomic propositions $p$ and $q$, respectively, the semantics of $\reventu\psi\land (\psi \releasedot (\psi\lor\varphi))$ can be expressed in terms of an LTL valuation W as:
\begin{align*}
	V(\sigma, \eventu q\land (q \release (q\lor p)) )=\left(W(\sigma, p \until q),~ W(\sigma, \eventu q),~ W(\sigma,\eventu q),~ W(\sigma,\eventu q)\right).
\end{align*}
Whereas $\varphi \untildot \psi$ only assumes two values, $\reventu\psi\land (\psi \releasedot (\psi\lor\varphi))$ assumes three possible values allowing to separate the words that violate $\varphi \until \psi$ into those that satisfy $\eventu q$ and those that do not.

\section{Relating LTL and $\mathrm{r}$LTL}
\label{subsec:rltl2ltl}
In this section we discuss, at the technical level, the relationships between $\rLTL(\P)$ and $\LTL(\P)$. Recall the mappings $\pi_j:\B_5\to\B$ introduced in~\eqref{Eq:Projection}, defined by \mbox{$\pi_j(a_1,a_2,a_3,a_4)=a_j$}, \mbox{$j\in\{1,2,3,4\}$}. Composing $\pi_j$ with the rLTL valuation $V$ we obtain the function $V_j=\pi_j\circ V$ that transforms an infinite word $\sigma\in \Sigma^\omega$ and an $\rLTL(\P)$ formula $\varphi$ into the element $V_j(\sigma,\varphi) \in \mathbb{B}$. Each truth value in $\B_5$ can be viewed as a sequence of 4 bits. Consequently, we show how to translate an $\rLTL(\P)$ formula $\varphi$ into four $\LTL(\P)$ formulae $\varphi_1, \ldots, \varphi_4$ such that:
\begin{align}
\label{eq:projVj}
	\pi_j(V(\sigma, \varphi)) = V_j(\sigma, \varphi) = W\left( \sigma, \varphi_j \right),
\end{align}
for all $\sigma \in \Sigma^\omega$ and $j \in \{1, \ldots, 4\}$. The key idea is to emulate the semantics of each operator occurring in $\varphi$ component-wise by means of dedicated LTL formulae. To make this clear and straightforward, we define the operator:
\begin{equation}
\label{eq:toLTLop}
\ltl : \{1, \ldots, 4\} \times \rLTL(\P) \rightarrow \LTL(\P),
\end{equation}
as in Table~\ref{table:toLTLop}. Then, each $\varphi_j$ formula is constructed as:
\begin{align}
\label{eq:eachbit}
	\varphi_j \coloneqq \ltl(j,\varphi) .
\end{align}

It is not hard to verify that the formulae $\varphi_j$ have indeed the desired meaning. Moreover, notice that due to the ordering of the rLTL truth values in $\B_5$, see \eqref{eq:ordering}, it follows that:
\begin{align}
\label{eq:toLTLordering}
	W(\sigma,\varphi_j) \geq W(\sigma,\varphi_i), ~\text{for $j \geq i$}.
\end{align}
Although the above construction only incurs a linear blowup in the number of subformulae of $\varphi_j$, the resulting LTL formula itself might be exponentially larger when explicitly constructed due to the recursive substitution.

\begin{table*}[t!]
\caption{The rLTL semantics via the $\ltl$ operator.}
\label{table:toLTLop}
\begin{center}
\bgroup
\def\arraystretch{}
\begin{tabular}{lcl}
\toprule
Operator & Symbol & Semantics for $p \in \P$, $\varphi, \psi \in \rLTL(\P)$.\\
\toprule
 Atomic Proposition  & & $\forall i \in \{1,2,3,4\} :  \ltl(i,p) = p$.
\\
\midrule
Negation & $\lnot$ & $\forall i \in \{1,2,3,4\} :  \ltl(i, \lnot \varphi) = \lnot \ltl(1, \varphi)$.
\\
\midrule
Conjunction & $\land$ & $\forall i \in \{1,2,3,4\} :  \ltl(i,  \varphi \land \psi) = \ltl(i, \varphi) \land \ltl(i, \psi)$.
\\
\midrule
Disjunction & $\lor$ & $\forall i \in \{1,2,3,4\} : \ltl(i,  \varphi \lor \psi) = \ltl(i, \varphi) \lor \ltl(i, \psi)$.
\\
\midrule
Robust Implication & $\rimplies$ & 
$
	\begin{aligned}
	 \forall i \in \{1,2,3\} : \, &\ltl(i, \varphi \rimplies \psi)  =  (\ltl(i,\varphi) \implies \ltl(i, \psi)) \land \ltl(i+1,\varphi \rimplies \psi), \\
	 				&\ltl(4, \varphi \rimplies \psi)  = (\ltl(4,\varphi)  \implies \ltl(4,\psi)). 
	\end{aligned}
$	
\\
\midrule
Next & $\nextdot$ & $\forall i \in \{1,2,3,4\} :  \ltl(i, \nextdot\varphi) = \Next \ltl(i,\varphi)$.
\\
\midrule
Robust Eventually & $\reventu$ & $\forall i \in \{1,2,3,4\} :  \ltl(i, \reventu \varphi) = \eventu \ltl(i,\varphi)$.
\\
\midrule
Robust Always & $\ralways$ &
$
	\begin{aligned}
	\ltl(1,\ralways \varphi)  & = \always \ltl(1, \varphi),  \\
	\ltl(2,\ralways \varphi)  & = \eventu\always \ltl(2, \varphi),\\
	\ltl(3,\ralways \varphi)  & = \always \eventu \ltl(3, \varphi),\\
	\ltl(4,\ralways \varphi)  & =  \eventu \ltl(4, \varphi).\end{aligned}
 $
\\ 
\midrule
Robust Until & $\untildot$ & $\forall i \in \{1,2,3,4\} :  \ltl(i, \varphi \untildot \psi) = \ltl(i, \varphi) \until \ltl(i,\psi)$.
\\
\midrule

Robust Release & $\releasedot$ & 
$
	\begin{aligned}
	\ltl(1, \varphi \releasedot \psi)  & = \ltl(1, \varphi) \release \ltl(1, \psi), \\
	\ltl(2, \varphi \releasedot \psi)  & = \eventu \ltl(2, \varphi) \lor \eventu \always \ltl(2, \psi), \\
	\ltl(3, \varphi \releasedot \psi) & = \eventu \ltl(3, \varphi) \lor \always \eventu  \ltl(3, \psi), \\
	\ltl(4, \varphi \releasedot \psi) &  = \eventu \ltl(4, \varphi) \lor \eventu \ltl(4, \psi) .
	\end{aligned}
$
\\
\bottomrule
\end{tabular}
\egroup
\end{center}
\end{table*}

Next, we show that the $\LTL(\P)$ semantics can be recovered from the first bit of the $\rLTL(\P)$ semantics. Specifically, for any $\phi \in \LTL(\P)$, there exists a $\varphi \in \rLTL(\P)$, such that \mbox{$V_1(\sigma, \varphi) = W(\sigma, \phi)$} for every $\sigma \in \Sigma^\omega$. The construction of a suitable $\varphi \in \rLTL(\P)$, given any $\phi \in \LTL(\P)$, is straightforward given Table~\ref{table:toLTLop}. First, every implication $\psi_1 \implies \psi_2$ in $\phi \in \LTL(\P)$ is replaced by the equivalent in LTL formula $\neg \psi_1 \lor \psi_2$. Second, $\phi$ is brought into negation normal form. Finally, by taking the robust version of the latter LTL formula, i.e., replacing the LTL temporal operators with their corresponding rLTL temporal operators, 
one obtains the desired $\varphi \in \rLTL(\P)$. 

\begin{remark}
	In the construction above, we replace every implication in the LTL formula $\phi$ by its LTL-equivalent negation form. 
	This is justified as, given the semantics of rLTL and LTL respectively, we have that for any infinite word $\sigma \in \Sigma^{\omega}$ the first bit of the rLTL valuation, $V_1(\sigma, \varphi \rimplies \psi)$, and the LTL valuation, $W(\sigma,\varphi_1 \implies \psi_1)$, are related as follows:
\begin{align*}
	V_1(\sigma, \varphi \rimplies \psi)  \leq V_1(\sigma,\neg \varphi \lor \psi) = W(\sigma,\neg \varphi_1 \lor \psi_1) = W(\sigma,\varphi_1 \implies \psi_1).
\end{align*}
Furthermore, we can construct simple examples for which the inequality above is strict, i.e., it fails to be an equality. For instance, consider the LTL formula $\always p \implies \always q$, its corresponding rLTL formula \mbox{$\ralways p \rimplies \ralways q$}, and the word \mbox{$\sigma =  \emptyset \{ p\}^\omega$}. Then, on one hand it holds that  \mbox{$V_1(\sigma, \ralways p \rimplies \ralways q) = 0$}, but on the other hand we have that \mbox{$V_1(\sigma, \neg \ralways p \lor \ralways q) = 1 = W(\sigma, \always p \implies \always q)$}. 
\end{remark}

The preceding discussion leads to the following result. 
\begin{lemma}
\label{lem:rLTLdecisionprobs}
Any rLTL formula $\varphi \in \rLTL(\P)$ can be translated to four LTL formulae $\varphi_j \in \LTL(\P)$, $j = 1, 2, 3, 4,$ as in \eqref{eq:eachbit} and such that \eqref{eq:projVj} holds, with $|\varphi_j | \leq c  |\varphi|$, for some $c > 0$. 
Moreover, for any LTL formula $\phi \in \LTL(\P)$ one can construct an rLTL formula $\varphi \in \rLTL(\P)$, such that $V_1(\sigma, \varphi) = W(\sigma, \phi)$ for every $\sigma \in \Sigma^\omega$, with $|\varphi| = |\phi|$. The aforementioned translations imply decidability of any problem for $\rLTL(\P)$, whose corresponding problem for $\LTL(\P)$ is decidable. Moreover, due to their effective translations to each other, $\LTL(\P)$ and $\rLTL(\P)$ are equally expressive. 
\end{lemma}
\begin{proof}
	Translating $\phi \in \LTL(\P)$ to $\varphi \in \rLTL(\P)$ only involves replacing the LTL temporal operators with their rLTL counterparts, which implies that $|\varphi| = |\phi|$. 
	The translation of $\varphi \in \rLTL(\P)$ to $\varphi_j \in \LTL(\P)$, $j = 1, 2, 3, 4$, is provided in \eqref{eq:eachbit}. This translation also results in a linear increase in the number of the unique subformulae, i.e., $|\varphi_j | \leq c  |\varphi|$, and the exact coefficient $c > 0$ is computed in Appendix~\ref{sec:appendixRLTL2LTLcomplexity}.
\end{proof}

\section{The \lowercase{r}LTL model checking problem}
\label{sec:rLTLmodelchecking}
Similarly to LTL, rLTL gives rise to various decision problems. One of the most important problems is the \emph{model-checking} problem. In this section we first formalize this problem in the context of rLTL. Then, through the translation from rLTL to LTL provided in Section~\ref{subsec:rltl2ltl}, the decidability of the rLTL model-checking problem is immediately settled. However, the treatment provided by the translation results in the construction of relatively large automata, which is inefficient for the purposes of model-checking. Therefore, we review the tighter known complexity bounds for the reader's convenience. Although the proof of such results is outside of the scope of this paper (the interested reader is referred to~\cite{tabuada2015rLTLarxiv, tabuada2016rLTL} for more details), they provide the setting in which we can appreciate the results in Section~\ref{sec:augmentedfragment} on an rLTL fragment for which the model checking problem can be solved with lower complexity.

\subsection{rLTL model-checking: definitions and background}
\label{subsec:rLTLproblems}

As discussed in Section~\ref{sec:prelims}, given a model of a system that is represented as a GBA, the LTL model-checking problem essentially asks whether or not all possible computations of the model satisfy an LTL specification. 
In a similar manner, the \emph{rLTL model-checking problem} is intuitively understood as the question of ``to what degree does a model satisfy a specification?''. The specification in this section is represented by an rLTL formula. For simplicity, we consider the model of the system to be given directly as a GBA. This leads to the following formulation of the \emph{rLTL model-checking problem}.

\begin{problem}[$\mathrm{r}$LTL model-checking]
\label{prob:rLTLmc}
	Given a set of atomic propositions $\P$, a GBA $\G$ with the corresponding set of words $L(\G) \subseteq \left (2^\P \right)^\omega$ that it recognizes, an rLTL formula $\varphi \in \rLTL(\P)$, and a truth value $b \in \B_5$, we ask if $V(\sigma,\varphi) \geq b$ holds for all $\sigma \in L(\G)$.
\end{problem}
In practice, one would be more interested in finding what is the largest $b \in \B_5$ such that $V(\sigma,\varphi) \geq b$ holds for all $\sigma \in L(\G)$. Section~\ref{subsec:fullrLTLbounds} provides the answer to Problem~\ref{prob:rLTLmc}, and repeatedly solving this problem for decreasing values of $b$ addresses the optimization problem of finding the largest $b \in \B_5$ such that $V(\sigma,\varphi) \geq b$ for all $\sigma \in L(\G)$. Note that this requires at most four invocations of the rLTL model-checking procedure and, hence, does not change the asymptotic complexity of the problem. 


\subsection{Improved bounds for rLTL model-checking}
\label{subsec:fullrLTLbounds}

In the preceding section, we provided a simple means to translate rLTL formulae into LTL formulae via the $\ltl$ operator in \eqref{eq:toLTLop}. Hence, as a consequence of Lemma~\ref{lem:rLTLdecisionprobs} the decidability question for Problem~\ref{prob:rLTLmc} is settled. In practice, however, this translation involves a blow-up, which results in relatively large automata that would then be used to solve the rLTL model-checking problem.

To alleviate this problem, a more efficient approach, via a direct translation from $\rLTL(\P)$ formulae into Generalized B\"uchi Automata (GBAs), was presented in~\cite{tabuada2016rLTL} for the $\rLTL_{\ralways,\reventu}$ fragment and for the full rLTL in~\cite{tabuada2015rLTLarxiv}. This construction, given an $\rLTL(\P)$ formula $\varphi$, results in a GBA with $\bigO \left(5^{|\varphi|} \right)$ states, where ${|\varphi|}$ is the number of subformulae of $\varphi$. This is the same complexity as for the LTL translation, which results in an automaton with size in $\bigO \left(2^{|\varphi|} \right)$, once we replace $2$ with $5$ since rLTL is 5-valued while LTL is 2-valued. 
Recall that in this paper, we reason about the model-checking complexity by focusing on the size of the corresponding GBA constructed from a given formula. 
On these grounds, we recall the following result from~\cite{tabuada2016rLTL, tabuada2015rLTLarxiv}. 

\begin{theorem}
\label{thm:rLTLmc}
Given an $\rLTL(\P)$ formula $\varphi$, the rLTL model-checking problem (Problem~\ref{prob:rLTLmc}) can be decided by using an automaton with at most:
\begin{align}
\label{eq:rLTLmcComplexity}
	\bigO \left( 5^{|\varphi|} \right) 
\end{align}
states, where $|\varphi|$ denotes the length of formula $\varphi$, i.e., the number of its distinct subformulae.
\end{theorem}

Interestingly, the many valued semantics of rLTL allows posing useful optimization problems in system verification. For instance, as discussed in the beginning of this section, one might be interested in the largest value that a system guarantees. Therefore, by repeatedly solving Problem~\ref{prob:rLTLmc} for decreasing truth values, one finds the largest value that a system guarantees. Theorem~\ref{thm:rLTLmc} provides a non-trivial upper bound, when compared to the direct translation via the $\ltl$ operator, as the following example suggests. 

\begin{example}
\label{ex:nontrivialbound}
Consider the $\rLTL(\P)$ formula $\ralways p \rimplies \ralways q$, and assume we wish to check $V(\sigma, \varphi) = 0111$ for every $\sigma \in L(\G)$, where $\G$ is the GBA representing the model of a given system. By using the $\ltl$ operator in Table~\ref{table:toLTLop}, this is equivalent to model-checking the formula:
\begin{align*}
	\ltl(2, \ralways p \rimplies \ralways q) & = 
	 (\always \eventu p \implies \always \eventu q) \land (\eventu \always p \implies \eventu \always q) \land ( \eventu p \implies \eventu q).
 \end{align*}
The original rLTL formula is of length $5$, and the LTL formula above has length $15$. Hence, by  na\"ively applying LTL complexity bounds \eqref{eq:LTLComplexity}, we count an upper bound of $2^{15}$ states the corresponding GBA, which is worse than the upper bound of $5^5$ states from \eqref{eq:rLTLmcComplexity}.
\end{example}

\section{A fragment for efficient \lowercase{r}LTL model-checking}
\label{sec:augmentedfragment}
Theorem~\ref{thm:rLTLmc} states that the $\rLTL(\P)$ model-checking problem can be solved by using a GBA of size $\bigO\left(5^{|\varphi|}\right)$. Nonetheless, this bound is still more expensive than the one of the $\LTL(\P)$ model-checking problem, which is $\bigO\left(2^{|\varphi|}\right)$. It is then natural to ask the following question, 
can we find a fragment of $\rLTL(\P)$ for which the model-checking problem can be solved more efficiently?

To motivate an answer to the above question, we show that the high complexity of $\rLTL(\P)$ model-checking stems from the fact that the four bits of an rLTL truth value are \emph{coupled} by the $\rimplies$ and $\neg$ operators. 
\begin{example}
\label{eg:expandbit}
Consider the $\rLTL(\P)$ formula $\varphi$ given by $\lnot(p \rimplies (q \releasedot r))$, where $p$, $q$ and $r$ are atomic propositions. To compute, for example, the $4$-th bit of its valuation, one needs to unfold the corresponding $\LTL(\P)$ formula:
\begin{align*}
	\ltl(4, \lnot(p \rimplies (q \releasedot r)) ) & = \lnot \ltl(1, p \rimplies (q \releasedot r)) = \lnot  \Big( ( \ltl(1,p) \implies \ltl(1, q \releasedot r) ) \land \ltl(2, p \rimplies  (q \releasedot r) \Big) \\ 
	& =  \neg \Big( (p \implies ( q \release r )) \land  \ltl(2, (p \rimplies (q \releasedot r) )) \Big)\\
	&= \big(p \land \lnot ( q \release r ) \big) \lor \lnot  \Big( \big( \ltl(2,p) \implies \ltl(2, q \releasedot r) \big) \land \ltl(3, p \rimplies  (q \releasedot r)) \Big) \\
	&= \big(p \land \lnot ( q \release r ) \big) \lor \lnot  \Big( \big( p \implies (\eventu q \lor \eventu\always r) \big) \land \ltl(3, p \rimplies  (q \releasedot r)) \Big) .
\end{align*}
If we continue unfolding the formula, we see that one needs to check an LTL formula that non-trivially depends on the bits 1, 2, and 3 of the valuation of $\varphi$. 
\end{example}

With this intuition in mind, we aim to identify a fragment of $\rLTL(\P)$ that suffers minimally from the coupling required by, e.g., the robust implication.

\subsection{Main results}
\label{subsec:mainresults}

We begin by defining the following fragments of $\rLTL(\P)$. 

\begin{definition}[Fragments $\thefragment$ and $\thelargerfragment$]
\label{def:rLTLtildaFinal}
Given a set of atomic propositions $\P$, define $\thefragment \subset \rLTL(\P)$ as the set of all rLTL formulae that either do not contain any $\rimplies$ operator, or if they do, the antecedent of the implication does not contain any $\ralways$ or $\releasedot$ operators. Then define the more general fragment $\thelargerfragment$ as:
\begin{align}
\label{eq:thebiggerfragment}
	\thelargerfragment = \thefragment \bigcup \left\{\psi_1 \rimplies \psi_2 \middle| \psi_1, \psi_2 \in \thefragment \right\},
\end{align}
which allows for one $\rimplies$ operator on the outermost level with no restrictions on the implication's antecedent.
\end{definition}

The main result of this section establishes refined complexity bounds for the model-checking problem of the above defined fragments.
\begin{theorem}
\label{thm:refinedboundsFinal}
Consider a set of atomic propositions $\P$.
The rLTL model-checking problem for any formula in $\thelargerfragment$ can be solved by performing at most $4$ LTL model-checking steps, each using an automaton with at most:
\begin{equation}
\label{eq:refinedcomplexityFinal}
	\bigO \left(  2^{|\varphi| - \kappa(\varphi)} 3^{\kappa(\varphi)} \right)
\end{equation}
states, where $\kappa(\varphi) = \mathrm{card}\left( \left\{ \psi \in \cl(\varphi) \mid \psi = \ralways \psi_1 \right\} \right) + \mathrm{card}\left( \left\{ \psi \in \cl(\varphi) \mid \psi = \psi_1 \releasedot \psi_2 \right\} \right)$ 
, i.e., the number of distinct subformulae of $\varphi$ of the form $\ralways \psi$ and $\psi_1 \releasedot \psi_2$, and $|\varphi|$ is the length of $\varphi$. 
\end{theorem}

\begin{remark}
Similarly to LTL, one can verify that for any $\varphi \in \rLTL(\P),$ it holds that $\reventu \varphi = \true \untildot \varphi$, and $\ralways \varphi = \false \releasedot \varphi$. Hence, the function $\kappa(\varphi)$ can be understood as the number of distinct $\psi_1 \releasedot \psi_2$ subformulae of $\varphi$, accounting for those underlying the $\ralways$ operators. 
\end{remark}

The above result provides a smaller complexity bound for the rLTL model-checking problem, which is closer to the tight bound of the LTL model-checking problem, i.e., $\bigO\left(2^{|\varphi|}\right)$. In fact, in the absence of $\ralways$ and $\releasedot$ operators, the bound in \eqref{eq:refinedcomplexityFinal} reduces to that of LTL. A special case of Theorem~\ref{thm:refinedboundsFinal} is the result proposed in~\cite{anevlavis2018rLTL}, where the same complexity bound is proven for a smaller rLTL fragment, specifically the set of rLTL formulae that do not contain any $\releasedot$ or $\rimplies$ (regardless of their antecedent) operators, or that allow for at most one $\rimplies$ operator only on the outermost level. 

\begin{remark}
Throughout this section we compare the model-checking complexity of the proposed fragment $\thelargerfragment$ to that of full LTL. This is justified by observing that for any LTL formula $\phi \in \LTL(\P)$, after replacing every implication $\psi_1 \implies \psi_2$ therein by  $\neg \psi_1 \lor \psi_2$, the corresponding rLTL formula belongs to the proposed fragment. This observation follows directly by \eqref{eq:thebiggerfragment}. 
\end{remark}

To illustrate the practicality of the proposed fragment $\thelargerfragment$, we examine the LTL formulae found in the B\"uchi Store~\cite{buchistore}, an open repository of LTL formulae. More than $95\%$ of the corresponding rLTL versions of these formulae belong to $\thelargerfragment$. Moreover, as discussed in Section~\ref{subsec:reactivitypatterns} in more detail, all the relevant reactivity patterns~\cite{dwyer1999patterns} fit into the fragment $\thelargerfragment$. Finally, the proposed fragment includes the GR(1) fragment. These considerations establish $\thelargerfragment$ as a practically useful fragment of rLTL. 

\begin{algorithm}[t]
\caption{Model-checking algorithm for the $\thelargerfragment$ fragment.}
\label{alg:earlystopping}
{\bf{Input: }}{A language $L(\G)$ generated by a GBA $\G$, and a formula $\varphi \in \thelargerfragment$.}\\
{\bf{Output: }}{The truth value of $\varphi$ on $L(\G)$, i.e., $b(L(\G), \varphi) \in \B_5$. }\\ 
 \For{$j = 0, 1, 2, 3$}
 {
 	$w := \inf_{\sigma \in L(\G)} W(\sigma, \ltl(4 - j, \varphi)).$\\
 	\If{$w = 0$}
 	{
 		\Return $\bbit[j]$
 	}
 }
 \Return $\bbit[4]$ 
\end{algorithm}

We devote the rest of this section to formally proving Theorem~\ref{thm:refinedboundsFinal}. Towards this, we introduce Algorithm~\ref{alg:earlystopping}. 
At a high level, Algorithm~\ref{alg:earlystopping} translates an rLTL formula into four LTL formulae and then model-checks each of them. The key idea is that when operating within the proposed fragment, the four formulae are independent of each other. It is actually this independence that allows us to construct smaller automata, by the use of temporal testers, and attain the desired complexity bounds. By exploiting the ordering of the five truth values in~\eqref{eq:ordering} the algorithm stops once an LTL formula is falsified. 

Overall in this section we prove the following points: (1) the fragment $\thefragment$ allows independent bit-wise computations for the rLTL model-checking problem; (2) smaller automata can be constructed for the formulae in the $\thefragment$ fragment by utilizing temporal testers; and (3) for any formula in the $\thelargerfragment$ fragment, the claimed bound on the size of the corresponding automaton is satisfied. 

\subsection{Bit-wise independence of $\thefragment$ and insights on the robust implication}
\label{subsec:robustnessinsights}
The goal of this section is to establish that \emph{for any formula $\varphi \in \thefragment$, the truth value of the $j$-th bit of its valuation, $V_j(\sigma,\varphi)$, is independent of any bit $i \neq j$}. Recall from Definition~\ref{def:rLTLtildaFinal} that $\thefragment$ contains all rLTL formulae that either do not contain any $\rimplies$ operators or, if they do, there exist no $\ralways$ or $\releasedot$ operators in the implication's antecedent. 

For the formulae that do not contain any $\rimplies$ operators, the desired result is easily verified by inspection of Table~\ref{table:toLTLop}. For the formulae that contain $\rimplies$, but their assumptions do not contain any $\ralways$ or $\releasedot$ operators we perform the following analysis. 

As Examples~\ref{ex:nontrivialbound} and~\ref{eg:expandbit} show, the $\rimplies$ operator induces coupling of the different LTL formulae corresponding to the bits of an rLTL truth value. Thus, we now draw insight on when it makes sense for an implication to be evaluated robustly. The notion of robustness, namely that ``small violations of the assumption lead to, at most, small violations of the guarantee'', is infused in the semantics of the $\rimplies$ operator, see \eqref{eq:semanticsrimplies}, since by \emph{weakening} the assumption and the guarantee, different rLTL truth values are obtained. Consequently, we investigate what ``weakening a formula'' means in the context of rLTL. 

\begin{definition}[Weakened $\rLTL(\P)$ formula]
Given a formula $\varphi \in \rLTL(\P)$, we say that $\varphi$ admits a \emph{weakened version} if there exists an infinite word $\sigma \in \Sigma^\omega$ such that \mbox{$V(\varphi,\sigma)\in \{0001, 0011, 0111\}$}. 
\end{definition}

Based on the above definition, for any rLTL formula $\varphi$ that does not admit a weakened version and any infinite word $\sigma \in \Sigma^\omega$, we have that $V(\varphi,\sigma)\in \{0000,1111\}$, i.e., the valuation of $\varphi$ admits only a binary truth value. This is equivalent to the statement that the corresponding LTL formulae $\varphi_{j}$, for $j \in \{ 1, 2, 3, 4 \}$, defined in \eqref{eq:eachbit} are semantically equivalent. Given the rLTL semantics, we make the following crucial observation. 

\begin{proposition} 
\label{prop:weakformula}
	Given a formula $\varphi \in \rLTL(\P)$, if $\varphi$ admits a \emph{weakened version}, then $\varphi$ contains at least one $\ralways$ or one $\releasedot$ operator.
\end{proposition}
\begin{proof}
We prove the result by contraposition. 
Consider any $\varphi \in \rLTL(\P)$ that does not contain any $\ralways$ or $\releasedot$ operators. The proof proceeds in three steps. \\
\hspace*{0.5em}1) If, additionally, $\varphi$ does not contain any $\rimplies$ operators, by inspection of Table~\ref{table:toLTLop} we see that the corresponding LTL formulae $\varphi_{j}$, for $j \in \{ 1, 2, 3, 4 \}$, defined in \eqref{eq:eachbit} are identical. Thus, the valuations $W(\sigma, \varphi_j)$, $j \in \{ 1, 2, 3, 4 \}$, over any $\sigma \in \Sigma^\omega$ are equal, and, hence, $V(\sigma,\varphi)\in \{0000,1111\}$. \\
\hspace*{0.5em}2) If $\varphi$ is of the form $\phi \rimplies \psi$, where neither $\phi$ nor $\psi$ contain a $\rimplies$ operator, then from the semantics of robust implication in \eqref{eq:semanticsrimplies} we have that $V(\sigma,\varphi) = 1111$ if $V(\sigma,\phi) \preceq V(\sigma,\psi)$, and that $V(\sigma,\varphi) = V(\sigma,\psi)$ otherwise. Furthermore, from 1) $V(\sigma,\phi), V(\sigma,\psi)\in \{0000,1111\}$, and, hence, $V(\varphi,\sigma)\in \{0000,1111\}$. \\
\hspace*{0.5em}3) Finally, if $\varphi \in \rLTL(\P)$ does not contain any $\ralways$ or $\releasedot$ operators, but may contain an arbitrary number of $\rimplies$ operators, the same result follows simply by induction on the subformulae of $\varphi$. \\
Therefore, if $\varphi$ does not contain a $\ralways$, or a $\releasedot$ operator, then it does not admit a \emph{weakened version}. Equivalently, if $\varphi$ admits a weakened version, then $\varphi$ contains at least one $\ralways$, or one $\releasedot$ operator. This concludes the proof.
\end{proof}

We use this proposition to determine when it is necessary to evaluate an implication robustly. Given an rLTL formula of the form $\varphi \rimplies \psi$, if the assumption $\varphi$ does not admit a weakened version, then the implication does not have to be evaluated robustly, and the LTL equivalence $\neg \varphi \lor \psi$ can be used instead. The following proposition formalizes this idea. 
\begin{proposition}
\label{prop:implicationequivrLTL}
	An $\rLTL(\P)$ formula $\varphi \rimplies \psi$ is \emph{semantically equivalent} to $\neg \varphi \lor \psi$, i.e., for any $\sigma\in \Sigma^\omega$ it holds that $V(\sigma, \varphi \rimplies \psi)=V(\sigma, \neg \varphi \lor \psi)$ if $\varphi$ does not contain any $\ralways$ or $\releasedot$ operators. 
\end{proposition}

\begin{proof}
Consider the rLTL formula $\varphi \rimplies \psi$, where $\varphi$ does not contain a robust implication operator for simplicity. The valuation of $\varphi \rimplies \psi$ over $\sigma \in \Sigma^\omega$, denoted by $V(\sigma, \varphi \rimplies \psi)$, is equal to: 
\begin{align}
\label{eq:implicationLTLeval}
	\left( W\left( \sigma, \bigwedge_{j = 1}^4 (\varphi_j \implies \psi_j ) \right), W\left( \sigma, \bigwedge_{j = 2}^4 (\varphi_j \implies \psi_j ) \right), W\left( \sigma, \bigwedge_{j = 3}^4 (\varphi_j \implies \psi_j ) \right), W\left( \sigma, \varphi_4 \implies \psi_4 \right) \right),
\end{align}
where $\varphi_j, \psi_j$, for $j \in \{ 1, 2, 3, 4 \}$, are defined in \eqref{eq:eachbit} by using the $\ltl$ operator according to Table~\ref{table:toLTLop}. 
If $\varphi$ contains no $\ralways$ and no $\releasedot$, as described by Proposition~\ref{prop:weakformula}, the LTL formulae $\varphi_{j}$ are identical. Let us denote all of them by $\varphi_{1}$. Then $V(\sigma, \varphi \rimplies \psi)$ becomes: 
\begin{align}
\label{eq:implicationLTLnoweak}
	\left( W\left( \sigma, \bigwedge_{j = 1}^4 (\varphi_1 \implies \psi_j ) \right), W\left( \sigma, \bigwedge_{j = 2}^4 (\varphi_1 \implies \psi_j ) \right), W\left( \sigma, \bigwedge_{j = 3}^4 (\varphi_1 \implies \psi_j ) \right), W\left( \sigma, \varphi_1 \implies \psi_4 \right) \right),
\end{align}
where we can see that, contrary to \eqref{eq:implicationLTLeval}, all the antecedents are now identical. 
We now show that \eqref{eq:implicationLTLnoweak} is equal to:
\begin{align}
\label{eq:implicationLTLnoweak2}
	&\left(
	W\left( \sigma, \varphi_1 \implies \psi_1 \right) ,	~
	W\left( \sigma, \varphi_1 \implies \psi_2 \right) ,	~
	W\left( \sigma, \varphi_1 \implies \psi_3 \right) ,	~
	W\left( \sigma, \varphi_1 \implies \psi_4 \right) 
	\right).
\end{align}
For any $\sigma \in \Sigma^\omega$: \\ 
\hspace*{0.5em}1) Assume $W(\sigma,\varphi_1) = 0$. Then $W(\sigma,\varphi_1 \implies \psi_j) = 1, ~\text{for $j \in \{ 1, 2, 3, 4 \}$}$, and both \eqref{eq:implicationLTLnoweak} and \eqref{eq:implicationLTLnoweak2} are equal to $1111$.	\\
\hspace*{0.5em}2) Assume $W(\sigma,\varphi_1) = 1$, and denote each truth value in $\B_5$ as:
	\begin{align*}
		\B_5 & = \left\{0000,0001,0011,0111,1111\right\} = \left\{\bbit[0],\bbit[1], \bbit[2], \bbit[3], \bbit[4]\right\}.
	\end{align*}
	If $W(\sigma,\psi_j)=0$ for all $j \in \{1, 2, 3, 4\}$, then both \eqref{eq:implicationLTLnoweak} and \eqref{eq:implicationLTLnoweak2} are equal to $\bbit[0]$. Else, let $k$ be the smallest element of $\{1,2,3,4\}$ such that $W(\sigma,\psi_k)=1$. Then, by \eqref{eq:toLTLordering}, we have that \mbox{$W(\sigma,\psi_j) \geq W(\sigma,\psi_k)$} for $j \geq k$, yielding $W(\sigma,\psi_j) = 1$ for $j \geq k$, which, in turn, implies that both \eqref{eq:implicationLTLnoweak} and \eqref{eq:implicationLTLnoweak2} are equal to $\bbit[5-k]$. 
	
The above shows that \eqref{eq:implicationLTLnoweak} and \eqref{eq:implicationLTLnoweak2} are equal when evaluated over any $\sigma \in \Sigma^\omega$. By using the LTL equivalence between $\varphi_j \implies \psi_j$ and $\neg \varphi_j \lor \psi_j$, we have that \eqref{eq:implicationLTLnoweak2} is equal to:
\begin{align}
\label{eq:implicationLTLnoweak3}
	&\left(
	W\left( \sigma, \neg \varphi_1 \lor \psi_1 \right) ,	~
	W\left( \sigma, \neg \varphi_1 \lor \psi_2 \right) ,	~
	W\left( \sigma, \neg \varphi_1 \lor \psi_3 \right) ,	~
	W\left( \sigma, \neg \varphi_1 \lor \psi_4 \right) 
	\right)	,
\end{align}
and, thus, equal to \eqref{eq:implicationLTLnoweak}. This yields the desired equality $V(\sigma,\varphi \rimplies \psi) = V(\sigma,\neg\varphi\lor\psi)$. 

The result for any $\varphi \in \rLTL(\P)$ that does not contain any $\ralways$, or $\releasedot$ operators, but may contain an arbitrary number of $\rimplies$ operators, follows simply by induction on the subformulae of $\varphi$ since neither the assumptions, nor the guarantees of these implications contain any $\ralways$ or $\releasedot$ operators. 
\end{proof}

Finally, by Proposition~\ref{prop:implicationequivrLTL} and the rLTL semantics in Table~\ref{table:toLTLop}, we are able to conclude that for any \mbox{$\varphi \in \thefragment$} the value of $V_j(\sigma,\varphi)$ is independent of any bit $i \neq j$. This concludes the objective of this section. 

\subsection{Introducing temporal testers}
\label{subsec:temporaltesters}
In this section, we review the concept of temporal testers~\cite{pnueli2008temporaltesters,kesten1998algverltl}. Temporal testers are discrete transition systems equipped with justice conditions. One of the appeals of temporal testers is that they can be used to obtain automata recognizing infinite words that satisfy an LTL formula by composing testers recognizing its subformulae. For example, from testers for the formulae $p \until q$  and $\always r$, one can construct a tester for $p \until (\always r)$ by composing the testers for $p \until q$  and $\always r$ using the constraint $q=\always r$. We start by providing the necessary definitions. 

\begin{definition}
\label{def:temporaltesters}
	A temporal tester for an LTL formula $\varphi \in \LTL(\P)$ is a tuple $\T_{\varphi} = (S, \Theta, R, \mathcal{J})$ where:
	\begin{itemize}
		\item $S$ is the set of states, $S \subseteq \B^{ \cl(\varphi) }$. Each state $x \in S$ is a function $x : \cl(\varphi)\to \mathbb{B}$ mapping a formula $\psi \in \cl(\varphi)$ to an element of $\B$. We denote the evaluation of $x$ on $\psi$ as $x_\psi$ and interpret it as the truth value of $\psi$ at the state $x$.
		\item $\Theta \subseteq S$ is a set of initial states. 
		\item $R \subseteq S \times S$ is a transition relation. 
		\item $\mathcal{J} = \{ J_1, \dots, J_K \} \subseteq 2^S$ is the set of justice requirements, where $J \subseteq S$ for each $J \in \mathcal{J}$.
	\end{itemize} 
	A \emph{computation} of a tester is an infinite sequence of states $\gamma = x^{(0)} x^{(1)} \dots$ such that $x^{(0)} \in \Theta$, \mbox{$\left(x^{(i)},x^{(i+1)}\right) \in R$} for $i \geq 0$, and for every $J \in \mathcal{J}$, $\gamma$ contains infinitely many states $x^{(i)} \in J$. Given a computation $\gamma$, we let $\sigma(\gamma) \in \left(2^\mathcal{P}\right)^\omega$ be the word $\sigma(\gamma) = \sigma_0(\gamma) \sigma_1(\gamma) \dots$ where $\sigma_i(\gamma)$ is the subset of $\mathcal{P}$ defined by $p\in \sigma_i(\gamma)$ if and only if $x_p^{(t)}=1$.	
\end{definition}

It is easily verified that the above definition of a temporal tester is equivalent to the one in~\cite{pnueli2008temporaltesters,kesten1998algverltl}. 
\begin{example}[A tester for the until operator adapted from~\cite{pnueli2008temporaltesters}]
\label{ex:T_pUq}
A tester for $p \until q$, where $p$ and $q$ are atomic propositions, is as follows:
\begin{equation}
\label{eq:T_pUq}
	\T_{p \until q} : 
	\left \{ 
	\begin{aligned}
	S & = \left\{x^{(1)}, x^{(2)}, x^{(3)}, x^{(4)}, x^{(5)}\right\}, \\
	\Theta & = S, \\
	R &= \left\{ \left(x,x'\right) \in S \times S \mid x_{ p \until q}= x_q \lor \left(x_p \land x_{ p \until q}'\right) \right\}, \\
	\mathcal{J} &= \{J_1\}, ~ J_1 =  \left\{ x \in S \mid \lnot x_{ p \until q} \lor x_q = 1\right\}	.
	\end{aligned}
	\right .
\end{equation}
The above tester can be represented as an automaton with \emph{5 states}, see Figure~\ref{fig:T_pUq}. 
\end{example}

\begin{figure}[t!]

\subfigure[Tester $\T_{p \until q}$.]{\label{fig:T_pUq} \includegraphics[width=0.40\textwidth]{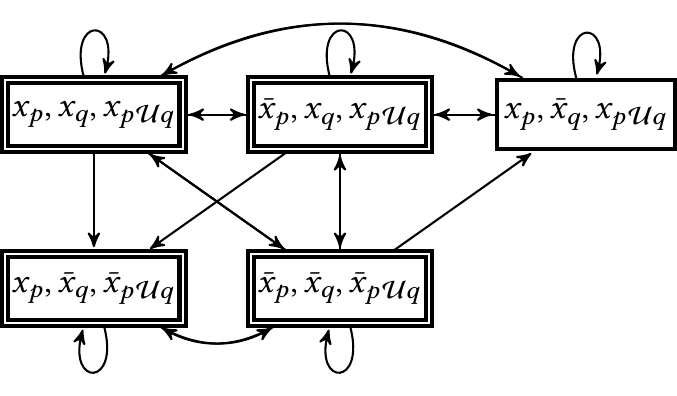}}
\subfigure[Tester $\T_{\eventu p}$.]{\label{fig:T_Fp} \includegraphics[width=0.40\textwidth]{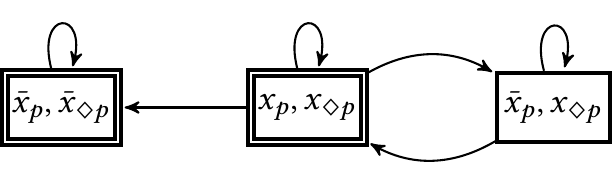}}
\subfigure[Tester $\T_{\always p}$.]{\label{fig:T_Gp} \includegraphics[width=0.40\textwidth]{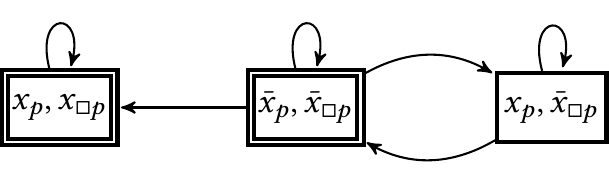}}
\subfigure[Tester $\T_{\eventu\always p}$ with two disjoint components.]{\label{fig:T_FGp} \includegraphics[width=0.40\textwidth]{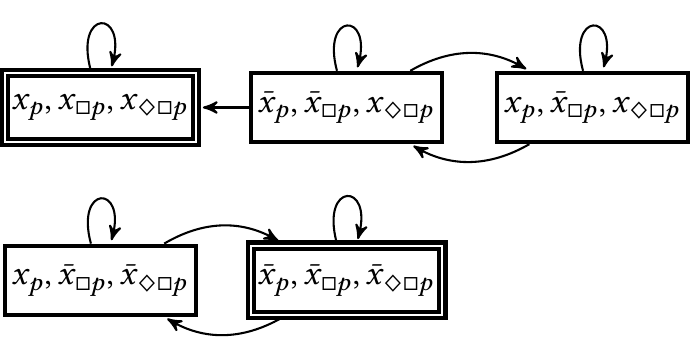}}

\caption{At each state, for a subformula $\psi$, $x_\psi$ denotes $x_\psi=1$, and $\bar{x}_\psi$ denotes $x_\psi = 0$. All states are valid initial states. The double line states are contained in the justice requirement.}
\end{figure}

In the general case, the tester $\T_{\varphi \until \psi}$, where $\varphi$ and $\psi$ are $\LTL(\P)$ formulae, is constructed by \emph{composition} of the testers for their subformulae. The following definition is adapted from~\cite[Section 3.2]{kesten1998algverltl},~\cite[Section 7]{pnueli2008temporaltesters}.

\begin{definition}[Composition of Temporal Testers]
\label{def:compositionJDS}
The synchronous parallel composition of two testers is $(S, \Theta, R, \mathcal{J}) = (S_1, \Theta_1, R_1, \mathcal{J}_1) \interleave (S_2, \Theta_2, R_2, \mathcal{J}_2),$ where $S = S_1 \cup S_2$, $\Theta = \Theta_1 \cap \Theta_2$, $R = R_1 \cap R_2$, and $\mathcal{J}  = \mathcal{J}_1 \cup \mathcal{J}_2$.
\begin{enumerate}
	\item For a unary LTL operator $\op$, a tester $\T_{\op(\varphi)}$ is given by $\T_{\op(p) \mid p \leftarrow \varphi}  \interleave \T_{\varphi},$ where the subscript $p \leftarrow \varphi$ denotes that we replace every instance of $x_p$ and $x_{\op(p)}$ in the first tester, by $x_{\varphi}$ and $x_{\op(\varphi)}$ respectively.
	\item For a binary LTL operator $\op$, a tester $\T_{\op(\varphi_1, \varphi_2)}$ is given by $ \T_{ \op(p,q) \mid p \leftarrow \varphi_1, q \leftarrow \varphi_2}  \interleave \T_{\varphi_1} \interleave \T_{\varphi_2}$, where we replace every instance of $x_p$, $x_q$ and $x_{\op(p,q)}$ in the first tester, by $x_{\varphi_1}$, $x_{\varphi_2}$ and  $x_{\op(\varphi_1, \varphi_2)}$ respectively.
\end{enumerate}
\end{definition}

By using the identities $\eventu p = \true \until p$, $\always p = \lnot \eventu \lnot p$, and $\T_{p \until q}$, we construct $\T_{\eventu p}$ and $\T_{\always p}$, which are shown in Figures~\ref{fig:T_Fp} and~\ref{fig:T_Gp}. By composing them, we obtain $\T_{\eventu \always p}$, shown in Figure~\ref{fig:T_FGp}. Such testers play an important role in proving smaller upper bounds for the rLTL model-checking problem. Towards this, the following results from~\cite{anevlavis2018rLTL} provide recursive bounds on the size of a tester $\T_{\varphi}$ for an $\LTL(\P)$ formula $\varphi$. 

\begin{definition}[Size of a tester]
Given a tester $\T_{\varphi}$ for $\varphi \in \LTL(\P)$, let $\lvert \T_{\varphi} \rvert $ denote its size, i.e., the number of its states, and let $\lvert \T_{\varphi} \rvert_{i}$ be the number of states where $x_\varphi = i$. Then, for any formulae $\varphi, \psi \in \LTL(\P)$, and $i,j,k \in \B$: 
\begin{itemize}
	\item for any unary operator $\op$, $\lvert \T_{\op(\varphi)} \rvert_{i,j}$ is the number of states where $x_{\varphi} = i, x_{\op(\varphi)} = j$,
	\item for any binary operator $\op$, $\lvert \T_{\op(\varphi_1, \varphi_2)} \rvert_{i,j,k}$ is the number of states where $x_{\varphi} = i$, $x_{\psi} = j$, $x_{\op(\varphi,\psi)} = k$.
\end{itemize}
\end{definition}
The number of states in a tester can be decomposed as follows for any $ \varphi, \psi \in \LTL(\P)$: 
\begin{align*}
\lvert \T_{\op(\varphi)} \rvert = \sum_{i,j} \lvert \T_{\op(\varphi)} \rvert_{i,j}, \quad \lvert \T_{\op(\varphi_1, \varphi_2)} \rvert = \sum_{i,j,k} \lvert \T_{\op(\varphi_1, \varphi_2)} \rvert_{i,j,k}.
\end{align*}

\begin{proposition}
\label{prop:testerStatesBounds}
Let $p,q$ be two atomic propositions in $\P$, $\psi_1,\psi_2 \in \LTL(\P)$ be two LTL formulae, and $\op$ denote an LTL operator. The following holds:
	\begin{align}
		\label{eq:numCompose1}
		\lvert \T_{\op(\psi_1)} \rvert_{i, j} &\leq \lvert \T_{\psi_1} \rvert_{i} \cdot \lvert \T_{\op(p)}\rvert_{i, j}, \\
		\label{eq:numCompose2}
		\lvert \T_{\op(\psi_1, \psi_2)} \rvert_{i,j,k} &\leq \lvert \T_{\psi_1}\rvert_i \cdot \lvert \T_{\psi_2} \rvert_j \cdot \lvert \T_{\op(p,q)} \rvert_{i,j,k}.
	\end{align}
\end{proposition}

\begin{corollary}[Recursive Bounds]
\label{cor:recbounds}
Consider a tester $\T_{\varphi}$ for $\varphi \in \LTL(\P)$. The following recursive bounds hold on its number of states $\vert \T_{\varphi} \rvert$: 
\begin{align}
	\label{eq:atompropT}
	&\text{if $\varphi$ is $p \in \P$}: |\T_{\varphi}| = 2,\\
	\label{eq:negT}
	&\text{if $\varphi$ is $\lnot \psi$}: \forall i,j: \lvert \T_{\neg \psi} \rvert_{i, j} = 	\begin{cases} 
									\lvert \T_{\psi} \rvert_{i}, \text{ if $i \neq j$}, \\ 
									0 \text{ otherwise},
									\end{cases},\\
	\label{eq:allT}
	&\text{if $\varphi$ is $\op(\psi), \ \op \in \{\eventu, \always, \Next\}$}:		\lvert \T_{\varphi} \rvert  \leq 2 \cdot \lvert \T_{\psi} \rvert, \\
	\label{eq:TEAp}	
	&\text{if $\varphi$ is $\eventu \always \psi$}: 	|\T_{\varphi}| = \lvert \T_{\eventu \always \psi} \rvert \leq 3 \cdot \lvert \T_{\psi} \rvert,	\\
	\label{eq:binT}
	&\text{if $\varphi$ is $\op(\psi_1,\psi_2), \ \op \in \{\lor, \land, \implies\}$}:	\lvert \T_{\varphi} \rvert \leq \lvert \T_{\psi_1} \rvert \cdot \lvert \T_{\psi_2} \rvert, \text{ and}	\\
	\label{eq:exotBinT}
	&\text{if $\varphi$ is $ \psi_1 \until \psi_2$ or $\varphi$ is $\psi_1 \release \psi_2$}:	\lvert \T_{\varphi} \rvert \leq 2\cdot \lvert \T_{\psi_1} \rvert \cdot \lvert \T_{\psi_2} \rvert.
\end{align}
\end{corollary}

Similarly to LTL, see Corollary~\ref{cor:LTLmc}, the complexity of the rLTL model-checking problem is proportional to the size of the GBAs constructed, see Theorem~\ref{thm:rLTLmc}. Hence, we are motivated to construct GBAs with the smallest possible number of states and accepting conditions. We focus on elementary temporal testers arising in the study of rLTL formulae and use them to construct smaller GBAs using the following remark. 

\begin{remark}[Link with Generalized B\"uchi Automata]
\label{rem:TesterToGBBA}
For any tester $\T_{\varphi} = (S, \Theta, R, \mathcal{J})$, one can construct a GBA $\G_\varphi = (Q, \Sigma, Q_0, \Delta, \mathcal{F})$ whose runs correspond to the computations of the tester as follows:
\begin{itemize}
	\item $Q = S$, $Q_0 = \left\{ x \in \Theta \mid x_\varphi = 1 \right\}$, 
	and $\mathcal{F} = \mathcal{J}$. 
	\item $\left(q,\sigma,q'\right)\in \Delta$ if and only if $\left(q,q'\right)\in R$ and $\sigma=\left\{p\in\mathcal{P}\,\,\vert\,\, q'_p=1\right\}$.
\end{itemize}
Notice the relation between $Q_0$ and $\Theta$. Since $\T_{\varphi}$ detects whether a computation satisfies $\varphi$ or $\lnot \varphi$, in order to obtain a GBA $\G_\varphi$ whose runs satisfy $\varphi$, it suffices to remove from $\Theta$ any state $x$ satisfying $x_\varphi = 0$.
\end{remark}

\begin{remark}[Justice requirements]
\label{rem:Justices}
In $\T_{p \until q}$, and therefore in $\T_{\eventu p}$, $\T_{\always p}$, the number of justice requirements $J \in \mathcal{J}$ is always 1. Following Definition~\ref{def:compositionJDS}, the composition $\T_{\eventu \always p} = \T_{\eventu p \mid p \leftarrow \always p} \interleave \T_{\always p}$ has two justice requirements:
\begin{align*}
	 \mathcal{J} = \{x \in S \mid \lnot x_{\eventu \always p} \lor x_{\always p} = 1\}  \cup \{x \in S \mid x_{\always p} \lor \lnot x_{p} = 1\}.
\end{align*}
The two justice requirements are met simultaneously at the states $(x_p, x_{\always p}, x_{\eventu \always p})$ and $(\bar{x}_p, \bar{x}_{\always p }, \bar{x}_{\eventu \always p})$. In light of this, we use $\mathcal{J} = \{x \in S \mid (x_p \land x_{\always p} \land x_{\eventu \always p}) \lor \lnot(x_p \lor x_{\always p} \lor x_{\eventu \always p})= 1\}$
 for the tester in Figure~\ref{fig:T_FGp}, while preserving the computations of $\T_{\eventu \always p} = \T_{\eventu p \mid p \leftarrow \always p} \interleave \T_{\always p}$. In this regard, the tester $\T_{\eventu \always p}$ in Figure~\ref{fig:T_FGp} is optimized. 
\end{remark}

We are now ready to proceed onto constructing smaller, specialized GBAs for the fragment $\thefragment$ and prove the refined complexity upper bounds. 
\subsection{Refined complexity bounds}
\label{subsec:refinedbounds}

The next lemma is integral to providing the promised complexity bounds of Theorem~\ref{thm:refinedboundsFinal}.

\begin{lemma}
\label{lem:fragmentbounds}
	Given a set of atomic propositions $\P$, for any $\varphi \in \thefragment$ and any \mbox{$j \in \{1, 2, 3, 4 \}$}:
	\begin{equation}
	\label{eq:testerBound}
		\lvert \T_{\ltl(j,\varphi)} \rvert \leq 2^{\vert \varphi \rvert - \kappa(\varphi)}3^{\kappa(\varphi)},
	\end{equation}
	where $\kappa(\varphi) = \mathrm{card}\left( \left\{ \psi \in \cl(\varphi) \mid \psi = \ralways \psi_1 \right\} \right) + \mathrm{card}\left( \left\{ \psi \in \cl(\varphi) \mid \psi = \psi_1 \releasedot \psi_2 \right\} \right)$ 
	, i.e., the number of distinct subformulae of $\varphi$ of the form $\ralways \psi$ and $\psi_1 \releasedot \psi_2$, and $|\varphi|$ is the length of $\varphi$. 
\end{lemma}

\begin{proof}
	To maintain a streamlined presentation we provide a sketch of the proof here. The full proof of Lemma~\ref{lem:fragmentbounds} is found in Appendix~\ref{sec:appendixEfficientFragment}.
	
	Given any $\varphi \in \thefragment$ we prove the claimed bound on $\lvert \T_{\ltl(j,\varphi)} \rvert$, $j \in \{1, 2, 3, 4 \}$,  by induction on the length of the formula for all operators except for the $\releasedot$ operator. In the base case where $\varphi$ is of length 1, i.e., an atomic proposition, \eqref{eq:testerBound} is satisfied as an equality given \eqref{eq:atompropT}. For the induction step, if $\varphi$ is of the form $\op(\psi)$, with $\op$ any unary rLTL operator, or of the form $\op(\psi_1, \psi_2)$, with $\op$ any binary operator except for $\releasedot$, then the claim in \eqref{eq:testerBound} is proved using \eqref{eq:negT} through \eqref{eq:exotBinT}.
	
	Finally, in the case of the $\releasedot$ operator we construct specialized testers for each bit, which we prove to be correct and to satisfy the claimed bound \eqref{eq:testerBound} in the Appendix.
\end{proof}


So far, Lemma~\ref{lem:fragmentbounds} shows that we can construct specialized temporal testers of smaller size. These, in turn, can be used to construct smaller GBAs as per Remark~\ref{rem:TesterToGBBA}. To conclude the proof of Theorem~\ref{thm:refinedboundsFinal}, we consider Algorithm~\ref{alg:earlystopping} for rLTL model-checking. This algorithm model-checks the LTL formulae corresponding to each bit of the truth value of an rLTL formula, and by exploiting the order of the truth values in $\B_5$, \mbox{$0000\prec 0001\prec 0011\prec 0111\prec 1111$}, benefits from an early stopping criterion if the value of a bit is zero. 
In particular, consider the language $L(\G)$ generated by a GBA $\G$, and assume we wish to model-check a formula $\varphi$ of the form $\psi_1 \rimplies \psi_2 \in \thelargerfragment$, as defined in \eqref{eq:thebiggerfragment}. Denote, again, each truth value in $\B_5$ as:
	\begin{align*}
		\B_5 & = \left\{0000,0001,0011,0111,1111\right\} = \left\{\bbit[0],\bbit[1], \bbit[2], \bbit[3], \bbit[4]\right\}.
	\end{align*}
Let $b(L(\G),\varphi) \in \B_5$ be the computed truth value. Algorithm~\ref{alg:earlystopping} first model-checks the corresponding LTL formula $\varphi_4$, which is of the form $\ltl(4,\psi_1) \implies \ltl(4,\psi_2)$. 
From Lemma~\ref{lem:fragmentbounds} and Remark~\ref{rem:TesterToGBBA}, this first model-checking step makes use of an automaton with the number of states as in \eqref{eq:refinedcomplexityFinal}. There are two possible outcomes: \\
\hspace*{0.5em}1) the formula is violated, i.e., $b(L(\G),\varphi) = \bbit[0]$, Algorithm~\ref{alg:earlystopping} terminates and Theorem~\ref{thm:refinedboundsFinal} holds; \\
\hspace*{0.5em}2) the formula is satisfied, i.e., $b(L(\G),\varphi) \succeq \bbit[1]$, and we need to check bit $3$. \\
However, having checked bit $4$, and given that neither $\psi_1$, nor $\psi_2$ contain a $\rimplies$ operator since they belong in $\thefragment$, it follows from the semantics in Table~\ref{table:toLTLop} that $V_3(\sigma, \varphi) = W(\sigma, \ltl(3,\psi_1) \implies \ltl(3,\psi_2))$. We can, hence, perform a second model-checking step, using an automaton with size, again, as in \eqref{eq:refinedcomplexityFinal}, and decide if $b(L(\G),\varphi) = \bbit[1]$ or $b(L(\G),\varphi) \succeq \bbit[2]$. The two other bits are computed similarly if needed. 

Overall, the number of model-checking steps that the algorithm goes through depends on $b(L(\G),\varphi)$. If we have exactly $\ell < 4$ bits set to $1$ in the valuation, then we need $\ell+1$ model-checking steps, i.e., to check that the bit $\ell$ is valued 1 and that the next bit is valued $0$. The fourth verification disambiguates between the values $\bbit[3]$ and $\bbit[4]$. Hence, if $b(L(\G),\varphi) = \bbit[\ell],$ we do $\min(\ell+1,4)$ model-checking steps. Consequently, by using Algorithm~\ref{alg:earlystopping}, the rLTL model-checking problem for any $\varphi \in \thelargerfragment$ is solved by performing at most $4$ LTL model-checking steps, each using an automaton with at most $\bigO \left(  2^{|\varphi| - \kappa(\varphi)} 3^{\kappa(\varphi)} \right)$ states. 
This concludes the proof of Theorem~\ref{thm:refinedboundsFinal} and this section.

\section{Case studies}
\label{sec:casestudies}
In the introduction we argued that system correctness is not sufficient for a good design as the system needs to also be robust. Towards this goal, we provide a series of case studies that exemplify the usefulness of rLTL when compared to standard LTL. 
For rLTL model-checking we use Evrostos{\footnote{Evrostos is available at: {\color{blue}https://github.com/janis10/evrostos}.}}\cite{anevlavis2019evrostos}. Evrostos is a tool for verifying rLTL specifications in the $\thelargerfragment$ fragment, and it is built on top of an LTL model-checker. Currently, it supports two modes: (1) {\tt{smv-mode}}, which uses NuSMV{\footnote{NuSMV is available at: {\color{blue}http://nusmv.fbk.eu}.}}~\cite{NuSMV} as the underlying model-checker; and (2) {\tt{pml-mode}}, which uses SPIN{\footnote{SPIN is available at: {\color{blue}http://spinroot.com}.}}~\cite{Spin}. 
In this section, we compare rLTL and LTL verification in terms of the ability to guarantee robustness, the information provided via verification, and the computational costs. 
Our case studies{\footnote{The files to replicate our case studies are available at: {\color{blue}https://github.com/janis10/evrostos/tree/master/case\_studies/}.}} 
show that model-checking in the proposed fragment $\thelargerfragment$:	\\
\hspace*{0.5em}1) identifies a non-robust system, which cannot be directly done with standard LTL; \\
\hspace*{0.5em}2) provides access to fine-grained information about the degree of specification violations, which can be useful towards improving the design;	\\
\hspace*{0.5em}3) incurs a relatively small computational overhead with respect to LTL model-checking; \\
\hspace*{0.5em}4) scales similarly to the LTL model-checking with respect to the size of the given formula, although slightly more expensive. 

\subsection{Case study 1: Aircraft Wheel Brake System (WBS)}
For our first case study, we revisit the aircraft WBS described in the Aerospace Information Report (AIR) 6110~\cite{AIR6110}. As aerospace systems have become more complex with the passing of time, it is essential that their development proceeds in a systematic way that minimizes errors. Towards this, the Federal Aviation Administration (FAA) specifies methods and guidelines~\cite{ARP4761,ARP4754A} for manufacturers, e.g., Boeing and Airbus, to guarantee that the development of products meets the necessary performance and safety requirements. 
AIR6110 provides an application of the specified processes~\cite{ARP4761,ARP4754A} to the 
example of a WBS. The WBS comprises a complex hydraulic plant managing two landing gears, each with four wheels, and controlled by an independent computer system. 

An extended formal verification of the WBS is found in~\cite{bozanno2015wbs} and the {\tt{.smv}} models used can be found in the references therein. The formal modeling and analysis are based on the integration of the contract-based design tool OCRA~\cite{OCRA}, the model-checker nuXmv~\cite{nuXmv}, and the xSAP platform for model-based safety analysis~\cite{XSAP}. 

\subsubsection{Formal specifications:}
A number of safety requirements from the AIR6110 document are formalized in LTL and are expressed as reactive specifications, where the assumption is on the environment of the WBS and the guarantee is the exact safety specification~\cite{bozanno2015wbs}. Here we focus on the requirement \emph{``S18-WBS-R-0325-wheelX: never inadvertent braking of wheel X without locking''}. The environment assumption for the safety specifications is that \emph{``at all times the power of the system is on, the power of the hydraulic pumps is on, and the hydraulic supplies maintain their nominal values''}. The above assumption and guarantee are formalized as a reactive LTL specification: 

\vspace{-2.5mm}
\begin{align}
\label{spec:wbs_reac_LTL}
	&\always \left( 
				\begin{aligned}	
					\bigwedge_{i=1}^{2} \tt{power_i} \bigwedge_{i=1}^{2} \tt{pump\_power_i} 
					\bigwedge_{i=1}^{2} \tt{hydraulic\_supply_i}=10
				\end{aligned} \right) \implies \\
	&\qquad\qquad\qquad\qquad\qquad\qquad\qquad 
	\always \neg \left(
				 \begin{aligned} 
				 		\neg&\tt{mechanical\_pedal_L} \, \land \, \tt{wheel\_status}=\tt{rolling} \\
						 &\land \, \tt{wheel\_braking\_force}>0 \, \land \, \tt{ground\_speed>0} 
				\end{aligned} \right) 	\nonumber
	,
\end{align}
where 
$\tt{power_i}$, $\tt{pump\_power_i}$, $\tt{hydraulic\_supply_i}$ are the $i$-th system's power, pump power (both boolean), and hydraulic supply (integer) respectively, $\tt{mechanical\_pedal_L}$ (boolean) is true if the left pedal is pressed, {\tt{ground\_speed}} (integer) is the aircraft's current speed relative to the ground, {\tt{wheel\_status}} is either rolling or stopped, and {\tt{wheel\_braking\_force}} (integer) is the force applied by the brakes to the wheel. 

The development of the WBS in the AIR6110 document is described through four evolutionary architectures:\\
\hspace*{0.5em}1) Arch1: 
comprises one Braking System Control Unit (BSCU) and one Hydraulic Circuit (HC) backed by an accumulator. \\
\hspace*{0.5em}2) Arch2: 
includes additional backup components, i.e., two BSCUs, a green HC, and a blue HC. \\
\hspace*{0.5em}3) Arch3: the two BSCUs of the control system are replaced by one dual channel BSCU. \\
\hspace*{0.5em}4) Arch4: accumulator placement is modified, a link from the control system validity to the selector valve in the physical system is added.
\\
For more details see~\cite{bozanno2015wbs}.

\subsubsection{Example scenario:}
We consider the Arch4 architecture and investigate how rLTL identifies a non-robust system, which cannot be directly done in LTL. While the nominal Arch4 model is correct and robust, we introduce a modification that makes it non-robust. In particular, we inject a bug in the sensor of the left pedal that makes it periodically miss the pressing of the pedal. This results in the BSCU not always receiving an electrical signal when the left pedal is pressed. 
To demonstrate our case, we consider a scenario where the environment assumption is violated finitely many times. This is reasonable as during the course of a flight, there can be perturbations to the environment assumption, but we expect the assumption to stabilize and not fail catastrophically. For example, the power input of the system might be interrupted, but eventually becomes stable. 

Model-checking separately the guarantee in \eqref{spec:wbs_reac_LTL} under the model with the sensor bug returns \false. Model-checking separately the assumption in \eqref{spec:wbs_reac_LTL} under the environment scenario above returns \false. However, model-checking the LTL specification in \eqref{spec:wbs_reac_LTL} under the discussed scenario returns \true. This is a consequence of the fact that in LTL, violation of the assumption leads to vacuous satisfaction of the specification. 

In contrast to LTL model-checking, we evaluate the corresponding rLTL reactive specification:
\begin{align}
\label{spec:wbs_reac_rLTL}
	&\ralways \left( 
				\begin{aligned}	
					\bigwedge_{i=1}^{2} \tt{power_i} \bigwedge_{i=1}^{2} \tt{pump\_power_i} 
					\bigwedge_{i=1}^{2} \tt{hydraulic\_supply_i}=10
				\end{aligned} \right) \rimplies \\
	&\qquad\qquad\qquad\qquad\qquad\qquad\qquad 
	\ralways \neg \left(
				 \begin{aligned} 
				 		\neg&\tt{mechanical\_pedal_L} \, \land \, \tt{wheel\_status}=\tt{rolling} \\
						 &\land \, \tt{wheel\_braking\_force}>0 \, \land \, \tt{ground\_speed>0} 
				\end{aligned} \right) 	\nonumber
	.
\end{align}
Using Evrostos, the resulting rLTL truth value is $0011$. This is interpreted as the guarantee being both satisfied and violated infinitely often under the environment of this scenario. To be more precise, the assumption is eventually always satisfied, meaning that its truth value is $0111$. However, the sensor pedal does not always pick up the pressing of the left pedal, meaning that it misses infinitely often during a system execution and, hence, the truth value of the guarantee is $0011$. By the semantics of robust implication in \eqref{eq:semanticsrimplies} we expect the truth value of the specification to be $0011$, which is what Evrostos returns. 

The above case study demonstrates the fact that the LTL implication cannot provide any actual information about the guarantee of a reactive specification whenever the assumption fails. Contrary, the rLTL implication does really verify whether a guarantee is satisfied, is violated, and to what degree.

\subsection{Case study 2: Meaningful reactivity rLTL patterns}
\label{subsec:reactivitypatterns}
For the second case study, we exhibit the practicality of the proposed fragment $\thefragment$ in Definition~\ref{def:rLTLtildaFinal}, and illustrate how rLTL provides more insight when a specification is violated compared to LTL. We first consider reactivity patterns of practical importance that occur commonly in the specification of concurrent and reactive systems~\cite{dwyer1999patterns}. Typical behaviors include the occurrence of a given event during system execution, such as absence, existence, and universality, or the relative order in which multiple events occur during system execution, such as precedence and response. The following corollary stems from studying all the relevant LTL patterns~\cite{dwyer1999patterns}.

\begin{corollary} 
	The relevant reactivity patterns~\cite{dwyer1999patterns} fall under the $\thefragment$ fragment, as described in Definition~\ref{def:rLTLtildaFinal}, when written in rLTL. 
\end{corollary}

It is the case for all these patters that the antecedent of any nested implication does not contain a $\always$ or a $\release$ operator. Hence, the same holds for their rLTL counterparts. By Proposition~\ref{prop:implicationequivrLTL}, we can equivalently write any robust implication formula $\varphi \rimplies \psi$ in these patterns, as $\neg \varphi \lor \psi$, which makes them immediately part of the efficient fragment $\thefragment$. We verified this for the 97 LTL formulae in~\cite{dwyer1999patterns}.

\begin{remark}
	A number of the patterns~\cite{dwyer1999patterns} are expressed using the ``weak until'' operator $\mathcal{W}$, which is related to the $\until$ operator by the LTL semantic equivalence between $p \mathcal{W} q$ and $p \until (q \lor \always p)$, and to the $\release$ operator by the LTL semantic equivalence between $p \mathcal{W} q$ and $q \release (q \lor p)$. It can be verified that in the context of rLTL, only the second equivalence captures the desired meaning of a robust version of the $\mathcal{W}$ operator. Therefore, to obtain the rLTL versions of these patterns, we first replace every subformula of the form $p \mathcal{W} q$ with $q \release (q \lor p)$.
\end{remark}

\subsubsection{Benchmark: Rigorous Examination of Reactive Systems (RERS)}
The second goal of this case study is to showcase how rLTL provides more fine-grained information when a specification is violated, as opposed to LTL, by mapping an LTL \false boolean value to different shades of false. To show this, we utilize the benchmarks found in the RERS Challenge~\cite{RERS}. The RERS Challenge contains a rich repository of problems of increasing complexity. 

Using available benchmarks, we analyze meaningful reactive specifications that fall under the patterns~\cite{dwyer1999patterns}, that is $160$ formulae spanning from RERS 2016 to RERS 2019. We use the provided {\tt{promela}} models from the RERS LTL parallel track and evaluate the rLTL counterparts of the specifications therein using Evrostos~\cite{anevlavis2019evrostos}. Our findings are summarized in Table~\ref{tab:rers}. 

\begin{table}[t!]
\caption{Occurrence frequency for each different rLTL truth value for $160$ properties from the RERS benchmark.}
\label{tab:rers}
\bgroup
\vspace{-2mm}
\begin{tabular}{cccccc}
\toprule
& \multicolumn{5}{c}{rLTL Truth Value} \\
\toprule
& $1111$ & $0111$ & $0011$ & $0001$ & $0000$ \\
\toprule
Frequency ($160$ formulae) & $53.75\%$ & $13.75\%$ & $25.625\%$ & $4.375\%$ & $2.5\%$	 \\
\midrule
\specialcell{Frequency ($70$ falsified formulae \\ that admit a weakened version)} & $-$ & $31.43\%$ & $58.57\%$ & $10\%$ & $0\%$	 \\
\bottomrule
\end{tabular}
\egroup
\end{table}

The considered set is balanced between falsified and satisfied specifications. Focusing on the falsified formulae, one appreciates the variation of the different shades of false that rLTL provides. More specifically, Table~\ref{tab:rers} indicates that, empirically, it is rarely the case that a falsified reactive specification fails catastrophically. Instead, a weaker version holds most of the times. In particular, when looking only at the formulae that do admit a weaker version according to Proposition~\ref{prop:weakformula}, the value $0000$ is not observed at all. To better interpret the above analysis, consider as an example the truth value $0111$. This value can be actually understood as ``the safety specification holds with a delay'', i.e., it is violated only finite times over any infinite trace of the system. This information can guide the designer towards more efficiently fixing the faulty model, so as to trace the root of the problem, or can provide insight about modifying a possibly inaccurate specification. 


\subsection{Case study 3: Studying the complexity blowup between LTL and rLTL}
We now aim to meaningfully compare the runtimes between LTL model-checking, and rLTL model-checking in the fragment $\thelargerfragment$. Towards this, we study the \emph{time complexity blowup}, $\zeta$, between LTL and rLTL, which is defined below. 

The complexity of the rLTL model-checking problem for any $\varphi \in \thelargerfragment$ , with respect to the GBA constructed for $\varphi$, is between $\bigO \left( 2^{| \varphi |} \right)$ and $\bigO \left( 3^{| \varphi |} \right)$ as Theorem~\ref{thm:refinedboundsFinal} establishes. 
Similarly, the complexity of the LTL model-checking problem for the corresponding LTL formula $\varphi_1$ is $\bigO \left( 2^{| \varphi_1 |} \right)$. Recall $\varphi_1$ is obtained as the LTL version of $\varphi$ simply by substituting the rLTL operators with their LTL counterparts.
Let the times required to solve the LTL and the rLTL model-checking problems be $t_{LTL}$ and $t_{rLTL}$ respectively. We know that $t_{LTL}$ is proportional to $2^{| \varphi_1 |}$, and notice that $| \varphi | = | \varphi_1 |$. Furthermore, $t_{rLTL} \geq t_{LTL}$, and, hence, we can write $t_{rLTL}=2^{\zeta | \varphi |}$, $\zeta \geq 1$. Then, we ask what is the exponent $\zeta$ that describes the overhead, i.e., what is the time complexity blowup. From the expressions above we obtain:
\vspace{-1.5mm}
\begin{align*}
	\zeta=1+\frac{\log_2 \Big( \frac{t_{rLTL}}{t_{LTL}} \Big)}{| \varphi |}	.
\end{align*}
Since the time complexity of rLTL model-checking for the proposed fragment $\thelargerfragment$ is proportional to at most $3^{| \varphi |}$, we have an upper bound for $\zeta$ of \mbox{$\log_{2}(3) = 1.58$}. 

\begin{figure}[t!]
\begin{minipage}{0.45\linewidth}
	\begin{tabular}{ccc}
	\toprule
	\multicolumn{3}{c}{rLTL Times (sec.)} \\
	\midrule
	$\min(t_{\rLTL})$ & $\max(t_{\rLTL})$ & $\text{mean}(t_{\rLTL})$ \\
	0.26 & 726.8 & 66.3 \\
	\bottomrule
	\multicolumn{3}{c}{LTL Times (sec.)} \\
	\midrule
	$\min(t_{\LTL})$ & $\max(t_{\LTL})$ & $\text{mean}(t_{\LTL})$ \\
	0.04 & 291.0 & 26.3 \\
	\bottomrule
	\multicolumn{3}{c}{Time Complexity Blowup} \\
	\midrule
	$\min(\zeta)$ & $\max(\zeta)$ & $\text{mean}(\zeta)$ \\
	1.016 & 1.122 & 1.058 \\
	\bottomrule
	\end{tabular}
\end{minipage}
\begin{minipage}{0.45\linewidth}
	\subfigure{\includegraphics{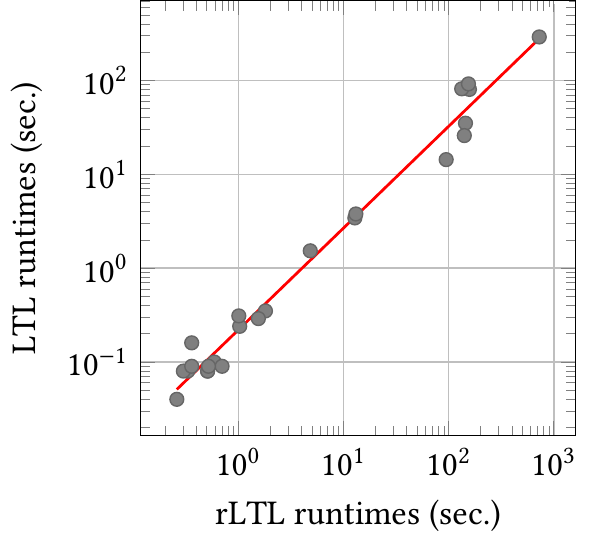}}
\end{minipage}
\vspace{-4mm}
\caption{Automated Air Traffic Control System. Left: minimum, maximum, and mean runtimes for rLTL and LTL model-checking, and time complexity blowup between rLTL and LTL. Right: Comparison between runtimes for rLTL and LTL (logarithmic scale). Computed over 24 experiments.} 
\label{fig:AAC}
\end{figure}

We use as a benchmartk the model of an Automated Air Traffic Control System~\cite{zhao2014acc}, designed for the Automated Airspace Concept (AAC), and the specifications therein. ACC is a high-level generic framework proposed as a candidate for the Next Generation Air Traffic Control System, which was under development at NASA. The goal of ACC is to always ensure the safe separation of commercial aircrafts within a given airspace sector, 
in order to prevent potential collisions. Figure~\ref{fig:AAC} shows on the left the minimum, maximum, and mean runtimes to model-check 12 safety and reactivity specifications, over 2 system models (original and abstract model) for a total of 24 experiments, and on the right how these runtimes are distributed. 
At a first glance, rLTL verification is somewhat more expensive computationally, which is expected as its runtime is proportional to at most $3^{| \varphi |}$ for a formula $\varphi \in \thelargerfragment$. However, observe that the time complexity blowup, $\zeta$, is well below this upper bound, even at its maximum value in this benchmark, meaning that the empirical time complexity of rLTL for the fragment we consider is close to that of LTL. 

\subsection{Case study 4:  Scalability}

\begin{figure}[t!]
\centering
\includegraphics[width=0.375\linewidth]{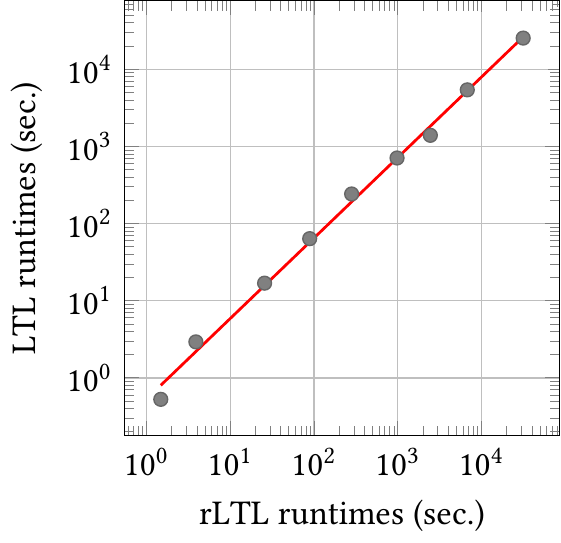}
\caption{Scalability comparison: model-checking runtimes for \eqref{eq:LTLeat1} and \eqref{eq:rLTLeat1} for $n = 1, \dots, 9$ (logarithmic scale).}
\label{fig:scalability}
\end{figure}

For our last set of experiments, we present a scalability case study for rLTL within the proposed fragment. To this end, we select the well-known model of the \emph{dining philosophers} and consider the following specification, $\phi$, saying that if whenever a philosopher is ready they eventually eat, then the first philosopher will eventually eat:
\begin{align*}
	\left( \bigwedge_{i=1}^n \always \eventu ready_i \implies \always \eventu eat_i \right) \implies eat_0.
\end{align*}
The above LTL formula contains nested implications with an $\always$ operator, which are not allowed in the proposed fragment $\thelargerfragment$. However, we can equivalently rewrite the LTL formula as:
\begin{align}
\label{eq:LTLeat1}
	\left( \bigwedge_{i=1}^n \neg (\always \eventu ready_i) \lor \always \eventu eat_i \right) \implies eat_0,
\end{align}
and then the corresponding rLTL specification $\varphi$ is:
\begin{align}
\label{eq:rLTLeat1}
	\left( \bigwedge_{i=1}^n \neg (\ralways \reventu ready_i) \lor \ralways \reventu eat_i \right) \rimplies eat_0. 
\end{align}
This specification presents a complex temporal structure as it contains multiple $\ralways$ and $\reventu$ operators, as well as an $\rimplies$ operator on the outermost level and, hence, is an appropriate candidate for our case study. We fix the model to have ten philosophers and then record the runtimes for evaluating \eqref{eq:LTLeat1} and \eqref{eq:rLTLeat1} for $n = 1, \dots, 9$. Our findings are summarized in Fig.~\ref{fig:scalability}. The specification $\varphi$ belongs to $\thelargerfragment$ and by Theorem~\ref{thm:refinedboundsFinal} can be model-checked using automata of size $\bigO \left( 2^{ |\varphi| - \kappa(\varphi)} 3^{\kappa(\varphi)} \right)$. We can appreciate that the runtimes for model-checking the rLTL specification scale in the same manner as those for the corresponding LTL model-checking, although slightly higher as expected by the aforementioned automaton size. This validates our theoretical results. This case study is also in agreement with the results presented in Fig.~\ref{fig:AAC} when studying the complexity blow-up between LTL and rLTL.

\section{Conclusion}
\label{sec:conclusion}
The logic rLTL provides a means to formally reason about both correctness and robustness in system design. While its syntax closely resembles that of LTL to ease its adoption, its semantics embeds a notion of robustness expressing that small deviations from the assumptions made at design time should lead to, at most, small violations of the design specifications. In this paper we presented a large fragment of rLTL, for which the verification problem can be efficiently solved by using an automaton of size $\bigO \left( 2^{ |\varphi| - \kappa(\varphi)} 3^{\kappa(\varphi)} \right)$, where $\kappa(\varphi)$ measures the number of unique subformulae of $\varphi$ that contain always and release operators. This bound is closer to the LTL bound of $\bigO \left( 2^{|\varphi|} \right)$ and an improvement to the previously known bound of $\bigO \left(5^{|\varphi|} \right)$. Moreover, at the time of publication there is no known non-trivial lower bound and finding such bound is an open problem. 
The usefulness of this fragment has been demonstrated by a number of case studies showing its expressiveness, the ability to capture robustness, and the benefits of the information the designer gains from the $5$-valued semantics towards refining the system and/or the specifications. Moreover, these advantages come at low computational overhead with respect to LTL model-checking, and a small learning curve from the designer as the syntax of rLTL closely mirrors that of LTL.

\bibliographystyle{alpha}
\bibliography{jMasterBib}

\appendix
\section{Proof of Lemma~\ref{lem:rLTLdecisionprobs}}
\label{sec:appendixRLTL2LTLcomplexity}
\begin{proof}
By using the $\ltl$ operator defined in Table~\ref{table:toLTLop} one translates an rLTL formula $\varphi \in \rLTL(\P)$ to four LTL formualae $\ltl(j,\varphi) \in \LTL(\P)$, $j=1, 2, 3, 4$, as defined in \eqref{eq:eachbit}. 
We show that the number of the unique subformulae resulting by the translation is linear to the size of $\varphi$, i.e.,:
\begin{align}
\label{eq:indgoal}
	\mathrm{card} \left( \bigcup_{j=1}^4 \cl \big( \ltl(j,\varphi) \big) \right) \leq c | \varphi |, 
\end{align}
for some $c \in \N$. We use an induction argument and show that \eqref{eq:indgoal} holds for $c=12$. 

\textbf{Base Case}. Take $\varphi \in \rLTL(\P)$ of length $1$, i.e., $\varphi$ is $p \in \P$. Then the claim holds straightforwardly as:
\mbox{$\mathrm{card} \left( \bigcup_{j=1}^4 \cl(\ltl(j,p)) \right) = \mathrm{card} \left( \bigcup_{j=1}^4 \cl(p) \right) = | p | \leq 12 | p |$}.

\textbf{Induction}. First, consider $\varphi \in \rLTL(\P)$ of the form $\psi_1 \releasedot \psi_2$ and note that $|\varphi| = \left\lvert \psi_1 \right\rvert + \left\lvert \psi_2 \right\rvert + 1$. Then:
\begin{align*}
	&\bigcup_{j=1}^4 \cl \big( \ltl(j,\varphi) \big) = \bigcup_{j=1}^4 \cl(\ltl(j,\psi_1 \releasedot \psi_2))  =
	\cl \big( \ltl(1, \psi_1) \release \ltl(1, \psi_2) \big) \cup \cl \big( \eventu\ltl(2, \psi_2) \lor \eventu\always\ltl(2, \psi_1) \big) \\
	&\hspace*{45mm}
	\cup \cl \big( \eventu\ltl(3, \psi_2) \lor \always\eventu\ltl(3, \psi_1) \big)  \cup \cl \big( \eventu\ltl(4, \psi_2) \lor \eventu\ltl(4, \psi_1) \big) \\
	&=
	\bigcup_{j=1}^4 \cl \big( \ltl(j,\psi_1) \big) \bigcup_{j=1}^4 \cl \big( \ltl(j,\psi_2) \big) 
	\cup \big\{ \ltl(1, \psi_1) \release \ltl(1, \psi_2) \big\} \\
	&\hspace*{40mm}
	\cup \big\{ \eventu\ltl(2, \psi_2), \always\ltl(2, \psi_1), \eventu\always\ltl(2, \psi_1), \eventu\ltl(2, \psi_2) \lor \eventu\always\ltl(2, \psi_1) \big\} \\
	&\hspace*{40mm}
	\cup \big\{ \eventu\ltl(3, \psi_2), \eventu\ltl(3, \psi_1), \always\eventu\ltl(3, \psi_1), \eventu\ltl(3, \psi_2) \lor \always\eventu\ltl(3, \psi_1) \big\} \\
	&\hspace*{40mm}
	\cup \big\{ \eventu\ltl(4, \psi_2), \eventu\ltl(4, \psi_1), \eventu\ltl(4, \psi_2) \lor \eventu\ltl(4, \psi_1) \big\} .	
\end{align*}
In turn the above implies that:
\begin{align*}
	\mathrm{card} \left( \bigcup_{j=1}^4 \cl \big( \ltl(j,\varphi) \big) \right) &\leq \mathrm{card} \left( \bigcup_{j=1}^4 \cl \big( \ltl(j,\psi_1) \big) \right) + \mathrm{card} \left( \bigcup_{j=1}^4 \cl \big( \ltl(j,\psi_2) \big) \right) + 12	\\
	&\leq 12 |\psi_1| + 12 |\psi_2| + 12 = 12 ( \left\lvert \psi_1 \right\rvert + \left\lvert \psi_2 \right\rvert + 1) = 12 |\varphi|, 
\end{align*}
where in the last inequality we used the induction hypothesis from \eqref{eq:indgoal} for $c = 12$. The other operators, with the exception of implication and negation, are similar and, thus, omitted for the sake of conciseness. 

We next consider $\varphi \in \rLTL(\P)$ of the form $\psi_1 \rimplies \psi_2$ and observe, given Table~\ref{table:toLTLop}, that:
\begin{align*}
	\cl \big( \ltl(j,\varphi) \big) \supseteq \cl(\ltl(i,\varphi)), ~ i \geq j, ~ i, j \in \{ 1, 2, 3, 4\}. 
\end{align*}
Then we have that:
\begin{align*}
	&\bigcup_{j=1}^4 \cl \big( \ltl(j,\varphi) \big) = \cl \big( \ltl(1,\varphi) \big) = \cl(\ltl(1,\psi_1 \rimplies \psi_2)) = \cl \left( \bigwedge_{j=1}^4 \ltl(j,\psi_1) \implies \ltl(j,\psi_2) \right) \\
	&=
	\bigcup_{j=1}^4 \cl \big( \ltl(j,\psi_1) \big) \bigcup_{j=1}^4 \cl \big( \ltl(j,\psi_2) \big) 
	\bigcup_{j=1}^4 \big\{ \ltl(j, \psi_1) \implies \ltl(j, \psi_2) \big\} \bigcup_{k=1}^3 \left\{ \bigwedge_{j=k}^4 \ltl(j, \psi_1) \implies \ltl(j, \psi_2) \right\} \\
	\Rightarrow &
	\mathrm{card} \left( \bigcup_{j=1}^4 \cl \big( \ltl(j,\varphi) \big) \right) \leq \mathrm{card} \left( \bigcup_{j=1}^4 \cl \big( \ltl(j,\psi_1) \big) \right) + \mathrm{card} \left( \bigcup_{j=1}^4 \cl \big( \ltl(j,\psi_2) \big) \right) + 7	\\
	&\leq 12 |\psi_1| + 12 |\psi_2| + 7 = 12 ( \left\lvert \psi_1 \right\rvert + \left\lvert \psi_2 \right\rvert + 1) = 12 |\varphi|, 
\end{align*}
The case of negation is similar as $\cl \big( \ltl(j,\varphi) \big) = \cl(\ltl(i,\varphi)), ~ i \neq j, ~ i, j \in \{ 1, 2, 3, 4\}$. This concludes the proof by induction. 

Finally, $| \ltl(j, \varphi) | \leq \mathrm{card} \left( \bigcup_{j=1}^4 \cl \big( \ltl(j,\varphi) \big) \right) \leq 12 | \varphi |$, $j \in \{1, 2, 3, 4\}$, which proves that the translation complexity is linear in the size of the formula, i.e., linear in the number of its unique subformulae, and concludes the proof. 
\end{proof}

\section{Proof of Lemma~\ref{lem:fragmentbounds}}
\label{sec:appendixEfficientFragment}
\begin{proof}
The proof follows from the rLTL semantics as defined in Table~\ref{table:toLTLop}, Proposition~\ref{prop:testerStatesBounds}, and Corollary~\ref{cor:recbounds}. We proceed by induction for all operators except for the $\releasedot$ operator, for which we construct a specialized tester.

\textbf{Base Case}. Take a formula $\varphi \in \thefragment$ of length $1$, i.e., $\varphi$ is $p \in \P$. Then we get \mbox{$\left\lvert \T_p \right\rvert = 2^{\left\lvert\varphi \right\rvert} = 2$}, which is the higher possible number of states here, and the claim holds.

\textbf{Induction}. First, consider formulae $\varphi \in \rLTL(\P)$ of the form $\op(\psi)$, where $\op$ is any unary rLTL operator. Note that $|\varphi| = |\psi| + 1$.	\\
\hspace*{0.5em}1) If $\op$ is $\ralways$, from \eqref{eq:TEAp} we obtain for $j \in \{1, 2, 3, 4 \}$: 
	\begin{align*}
                \left\lvert \T_{\ltl(j,\ralways \psi)} \right\rvert \leq 3\left\lvert \T_{\ltl(j,\psi)} \right\rvert \leq 3 \cdot 2^{\left\lvert \psi \right\rvert - \kappa(\psi)}3^{\kappa(\psi)} \leq  2^{\left\lvert \psi \right\rvert - \kappa(\psi) +1-1}3^{\kappa(\psi)+1}  = 2^{\left\lvert \varphi \right\rvert - \kappa(\varphi)}3^{\kappa(\varphi)} 
                ,
	\end{align*}
	which holds for the second and third bit.	\\
\hspace*{0.5em}2) If $\op$ is any other unary rLTL operator, from \eqref{eq:negT}, \eqref{eq:allT} we obtain for $j \in \{1, 2, 3, 4 \}$: 
	\begin{align*}
		\left\lvert \T_{\ltl(j,\op(\psi))} \right\rvert  \leq 2\left\lvert \T_{\ltl(i,\psi)} \right\rvert \leq  2^{\left\lvert \psi \right\rvert +1 - \kappa(\psi)}3^{\kappa(\psi)} = 2^{\left\lvert \varphi \right\rvert - \kappa(\varphi)}3^{\kappa(\varphi)}	
		,
	\end{align*}
	where $i = 1$ if $\op$ is $\lnot$, and otherwise $i=j$.

Consider now $\varphi \in \rLTL(\P)$ of the form $\op(\psi_1, \psi_2)$, where $\op$ is a binary operator. Its length is $|\varphi| = \left\lvert \psi_1 \right\rvert + \left\lvert \psi_2 \right\rvert + 1$.	\\
\hspace*{0.5em}1) If $\op$ is either $\land$ or $\lor$, from \eqref{eq:binT} we obtain for $j \in \{1, 2, 3, 4 \}$: 
	\begin{align*}
		\left\lvert \T_{\ltl(j,\op(\psi_1, \psi_2))} \right\rvert \leq \left\lvert \T_{\ltl(j,\psi_1)} \right\rvert \cdot \left\lvert \T_{\ltl(j,\psi_2)} \right\rvert \leq 2^{\left\lvert \psi_1 \right\rvert + \left\lvert \psi_2\right\rvert  - \kappa(\psi_1) - \kappa(\psi_2)} 3^{\kappa(\psi_1) + \kappa(\psi_2)} \leq 2^{\left\lvert \varphi \right\rvert - \kappa(\varphi)}3^{\kappa(\varphi)} 
		.
	\end{align*}
\hspace*{0.5em}2) If $\op$ is $\untildot$,  from \eqref{eq:exotBinT} we obtain for $j \in \{1, 2, 3, 4 \}$: 
	\begin{align*}
		\left\lvert \T_{\ltl (i,\psi_1 \untildot \psi_2)} \right\rvert \leq 2 \left\lvert \T_{\ltl(j,\psi_1)} \right\rvert \cdot \left\lvert \T_{\ltl(j,\psi_2)} \right\rvert \leq 2^{\left\lvert \psi_1 \right\rvert  + \left\lvert \psi_2 \right\rvert +1 - \kappa(\psi_1) - \kappa(\psi_2)} 3^{\kappa(\psi_1) + \kappa(\psi_2)} \leq 2^{\left\lvert \varphi \right\rvert - \kappa(\varphi)}3^{\kappa(\varphi)} 
		.
	\end{align*}
This concludes the proof by induction. 

Finally, we need to prove the same bounds for the $\releasedot$ operator. This case is slightly more tricky, but it suffices to show that for any \mbox{$\varphi, \psi \in \LTL(\P)$}, there exist appropriate testers such that:
\begin{align}
|\T_{\varphi \release \psi}| & \leq 3 \cdot |\T_{\varphi}| \cdot |\T_{\psi}|, \label{eq:bit1}\\
|\T_{\eventu \always \psi \lor \eventu \varphi}| & \leq 3 \cdot |\T_{\varphi}| \cdot |\T_{\psi}|, \label{eq:bit2} \\
|\T_{\always \eventu \psi \lor \eventu \varphi}| & \leq 3 \cdot |\T_{\varphi}| \cdot |\T_{\psi}|, \label{eq:bit3} \\
|\T_{\eventu \psi \lor \eventu \varphi}|  & \leq 3 \cdot |\T_{\varphi}| \cdot |\T_{\psi}|, \label{eq:bit4}
\end{align}
which by Table~\ref{table:toLTLop} are the testers for the corresponding $4$ LTL formulae due to the $\releasedot$ operator.	\\
\hspace*{0.5em}1) Proof of \eqref{eq:bit1}:  The inequality follows from \eqref{eq:exotBinT}. Notice that we can replace the constant $3$ by a $2$ here. \\
\hspace*{0.5em}2) Proof of \eqref{eq:bit2}: 
The LTL formulae $\eventu \always q \lor \eventu p$ and $\eventu ( \always q \lor p )$ are semantically equivalent for any $p, q \in \P$. Thus, we construct a tester for \mbox{$\eventu(\always q \lor p)$} in Figure~\ref{fig:bit2} by following the composition rules. This tester has at most 3 nodes that assign the same value to the atomic propositions $p$ and $q$: two states assign $(x_p, x_q, x_\varphi) = (0,1,1)$, and one state assigns $(x_p, x_q, x_\varphi) = (0,1,0)$. From this and Proposition~\ref{prop:testerStatesBounds} we conclude \eqref{eq:bit2}.	\\
\hspace*{0.5em}3) Proof of \eqref{eq:bit3}: We provide a \emph{specialized tester} for the formula $\varphi$ of the form $\always\eventu q \lor \eventu p$ in Figure~\ref{fig:bit3} and prove its correctness, i.e., show that the tester is both \emph{sound} and \emph{complete}~\cite[Section 5]{pnueli2008temporaltesters}. Recall from Definition~\ref{def:temporaltesters} that given a computation $\gamma$, we let $\sigma(\gamma) \in \left(2^\mathcal{P}\right)^\omega$ be the word $\sigma(\gamma) = \sigma_0(\gamma) \sigma_1(\gamma) \dots$ 
where $\sigma_t(\gamma)$ is the subset of $\mathcal{P}$ defined by $p \in \sigma_t(\gamma)$, if and only if, $x_p^{(t)}=1$. 
Soundness is defined as follows.

\begin{definition}[Soundness]
Given a formula $\varphi \in \LTL(\P)$, a tester $\T_{\varphi}$ is \emph{sound} if for all computations 
\mbox{$\gamma = x^{(0)} x^{(1)} \dots $ of $\T_{\varphi}$}, we have that $x^{(t)}_\varphi = 1$, if and only if, $W(\sigma(\gamma)_{t\dots}, \varphi) = 1$, where $\sigma(\gamma)_{t\dots}$ is the suffix of $\sigma(\gamma)$ starting at the $t$-th position. 
\end{definition} 

We prove soundness by considering all possible initial states for computations $\gamma = x^{(0)} x^{(1)} \dots$ of $\T_{\varphi}$:
\begin{itemize}
	\item If $x^{(0)}$ is any of the three states on the left, then at time $t=0$, we have \mbox{$W\left(\sigma(\gamma), \varphi \right) = 0$}.
	\item If $x^{(0)}$ is any of the two middle states, then $W(\sigma(\gamma), p) = 1$, hence $W(\sigma(\gamma), \varphi) = 1$.
	\item If $x^{(0)}$ is any of the rightmost states, then either the computations visit one of the middle states, or never do so. In the first case, we have $W(\sigma(\gamma), \eventu p) = 1$. In the second case, they had to visit the bottom right node infinitely often due to the justice requirements. Therefore, $W(\sigma(\gamma), \always \eventu q) = 1$. Combining the two cases results in $W(\sigma(\gamma), \varphi) = 1$.
 \end{itemize}
We conclude that the tester is sound and move onto proving its completeness.

\begin{figure}[t!]
	\subfigure[Tester $\T_{\eventu \always q \lor \eventu p}$.]{\label{fig:bit2} \includegraphics[width=0.45\textwidth]{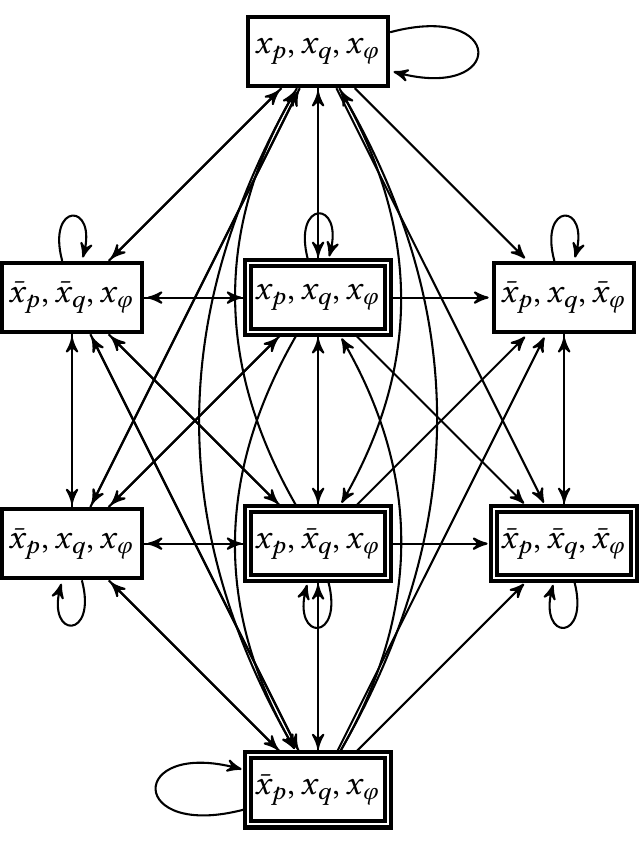}}
	\subfigure[Tester $\T_{\always \eventu q \lor \eventu p}$.]{\label{fig:bit3} \includegraphics[width=0.45\textwidth]{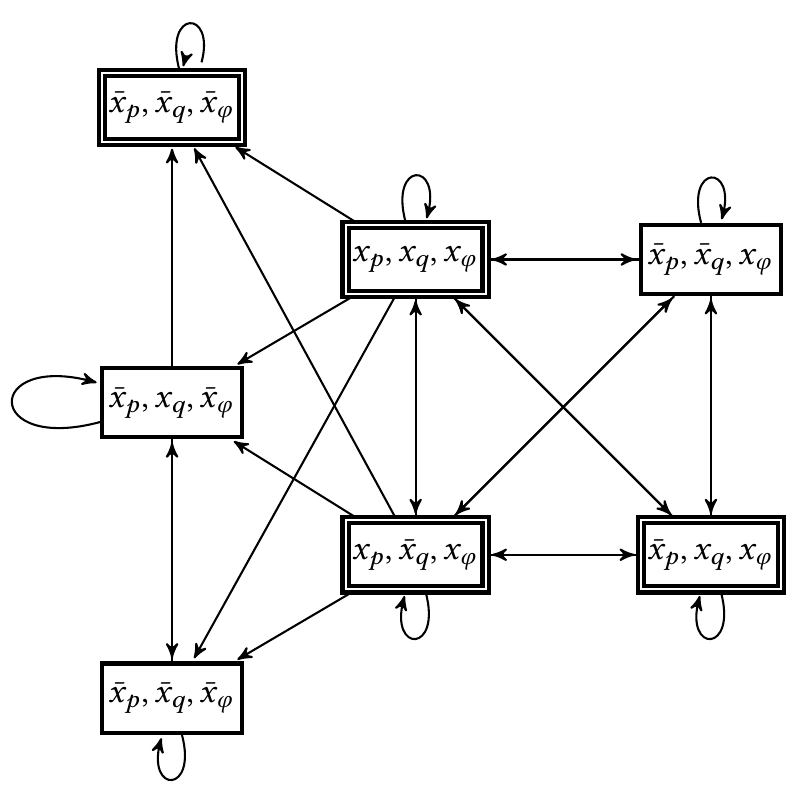}}
	\caption{All states are initial, double line states are contained in the set of justice requirements.}
\end{figure}

\begin{definition}[Completeness]
Given a formula $\varphi \in \LTL(\P)$, a tester $\T_{\varphi}$ is \emph{complete} if for any word $\sigma \in \left(2^\P\right)^\omega$, there exists a computation $\gamma$ of $\T$, such that 
$\forall t \geq 0$, $x^{(t)}_\varphi = 1$, if and only if, $W\left(\sigma_{t\dots}, \varphi\right) = 1$, where $\sigma_{t\dots}$ is the suffix of $\sigma$ starting at the $t$-th position. 
\end{definition} 

To prove completeness, we consider words  $\sigma \in \{p, q\}^\omega$ and different cases based on which subformulae of $\varphi$ they satisfy:
\begin{itemize}
	\item If $\sigma$ satisfies $\eventu p$, pick a computation with initial state either in the middle or in the right parts of the tester. The computation remains in these parts and visits middle states as many times as $p$ occurs in $\sigma$.  If it occurs infinitely often, then the computation satisfies the justice requirements. Otherwise, there is a $t > 0$ after which $\sigma_{t\dots}$ satisfies $\always \lnot p$, which is discussed next.
	\item If $\sigma$ satisfies $\always \lnot p$, we are only interested in the evolution of the $x_q$ variables. In particular, the corresponding computations are to remain either in the left or in the right parts of the tester. Hence, we can ignore the two middle nodes here. To conclude, observe that by doing so, we obtain a tester for $\always \eventu q$ (see  Figure~\ref{fig:T_FGp} and the fact that $\eventu \always \psi $ is equivalent to $\lnot \always \eventu \lnot \psi$).
\end{itemize}

The tester is therefore sound and complete. The constant $3$ in \eqref{eq:bit3} is obtained by counting the corresponding number of nodes in $\T_{\varphi}$ for all combinations of $x_p$ and $x_q$, and then using Proposition~\ref{prop:testerStatesBounds}. There are 3 nodes such that $x_p = x_q = 0$. \\
\hspace*{0.5em}4) Proof of \eqref{eq:bit4}: The proof here is direct when considering $\eventu p \lor \eventu q = \eventu(p \lor q) = \true \until(p \lor q)$ and equations \eqref{eq:binT} and \eqref{eq:exotBinT}.

This proves the bound for the $\releasedot$ operator and concludes the proof. 

\end{proof}

\end{document}